\documentclass[11pt,a4paper,UKenglish]{article}

\input{preamble}
\title{The Complexity Landscape of\\Fixed-Parameter Directed Steiner Network 
Problems\footnote{This work was supported by ERC Starting Grant PARAMTIGHT 
(No.~280152), ERC Consolidator Grant SYSTEMATICGRAPH (No.~725978), the Czech 
Science Foundation GA{\v C}R (grant \#19-27871X). A preliminary 
version of this paper~\cite{DBLP:conf/icalp/FeldmannM16} appeared in the 
proceedings of the 43nd International Colloquium on Automata, Languages, and 
Programming (ICALP), 
2016.}}
\date{}

\author[1]{Andreas Emil Feldmann}
\author[2]{D\'aniel Marx}
\affil[1]{Department of Applied Mathematics, Charles University, Prague,
Czechia. \texttt{feldmann.a.e@gmail.com}}
\affil[2]{CISPA Helmholtz Center for Information Security, Saarbrücken, Germany,
\texttt{marx@cispa.saarland}}

\begin{document}
\renewcommand*{\sectionautorefname}{Section}
\renewcommand*{\subsectionautorefname}{Section}
\renewcommand*{\algorithmautorefname}{Algorithm}

\maketitle

\begin{abstract}
Given a directed graph $G$ and a list $(s_1,t_1)$, $\dots$, $(s_d,t_d)$ of 
terminal pairs, the \pname{Directed Steiner Network} problem asks for a 
minimum-cost subgraph of $G$ that contains a directed $s_i\to t_i$ path for 
every $1\le i \le d$. The special case \pname{Directed Steiner Tree} (when we 
ask for paths from a root $r$ to terminals $t_1$, $\dots$, $t_d$) is known to be 
fixed-parameter tractable parameterized by the number of terminals, while the 
special case \pname{Strongly Connected Steiner Subgraph} (when we ask for a path 
from every $t_i$ to every other $t_j$) is known to be W[1]-hard parameterized by 
the number of terminals. We systematically explore the complexity landscape of 
directed Steiner problems to fully understand which other special cases are FPT 
or W[1]-hard. Formally, if $\mc{H}$ is a class of directed graphs, then we look 
at the special case of \pname{Directed Steiner Network} where the list 
$(s_1,t_1)$, $\dots$, $(s_d,t_d)$ of demands form a directed graph that is a 
member of $\mc{H}$. Our main result is a complete characterization of the 
classes $\mc{H}$ resulting in fixed-parameter tractable special cases: we show 
that if every pattern in $\mc{H}$ has the combinatorial property of being 
``transitively equivalent to a bounded-length caterpillar with a bounded number 
of extra edges,'' then the problem is FPT, and it is W[1]-hard for {\em every} 
recursively enumerable $\mc{H}$ not having this property. This complete 
dichotomy unifies and generalizes the known results showing that \pname{Directed 
Steiner Tree} is FPT [Dreyfus and Wagner, \textit{Networks} 1971], 
\pname{$q$-Root Steiner Tree} is FPT for constant $q$ [Such\'y, \textit{WG} 
2016], \pname{Strongly Connected Steiner Subgraph} is W[1]-hard [Guo et al., 
\textit{SIAM J. Discrete Math.} 2011], and \pname{Directed Steiner Network} is 
solvable in polynomial-time for constant number of terminals [Feldman and Ruhl, 
\textit{SIAM J. Comput.} 2006], and moreover reveals a large continent of 
tractable cases that were not known before.

\end{abstract}

\section{Introduction}

\pname{Steiner Tree} is a basic and well-studied problem of
combinatorial optimization: given an edge-weighted undirected graph
$G$ and a set $R\subseteq V(G)$ of terminals, it asks for a
minimum-cost tree connecting the terminals.  The problem is well known
to be NP-hard, in fact, it was one of the 21 NP-hard problems
identified by Karp's seminal paper~\cite{MR51:14644}.  There is a
large literature on approximation algorithms for \pname{Steiner Tree}
and its variants, resulting for example in constant-factor
approximation algorithms for general graphs and approximation schemes
for planar graphs
\citep[see][]{DBLP:journals/talg/BorradaileKM15,
DBLP:journals/talg/DemaineHK14,
DBLP:journals/jacm/ByrkaGRS13,DBLP:conf/icalp/BateniHL13,
DBLP:journals/algorithmica/BateniH12,DBLP:journals/jacm/BateniHM11,
DBLP:journals/siamcomp/ArcherBHK11,
DBLP:journals/talg/BorradaileKM09,
DBLP:journals/siamdm/RobinsZ05,DBLP:conf/soda/RajagopalanV99,
DBLP:journals/jal/KleinR95,DBLP:journals/siamcomp/AgrawalKR95,
DBLP:journals/networks/DreyfusW71}.
From the viewpoint of parameterized algorithms, the first result is the
classic dynamic-programming algorithm of Dreyfus and 
Wagner~\cite{DBLP:journals/networks/DreyfusW71} from 1971, which solves the
problem with $k=|R|$ terminals in time $3^k\cdot n^{O(1)}$. This shows that the 
problem is fixed-parameter tractable~\cite{pc-book,downey2013fpt} (FPT) 
parameterized by the number of terminals, i.e., there is an algorithm to solve 
the problem in time $f(k)\cdot n^{O(1)}$ for some computable function~$f$. In 
this paper we will only be concerned with this well-studied parameter $k=|R|$. A 
more recent algorithm by Fuchs et al.~\cite{fuchs2007dynamic} obtains runtime 
$(2+\delta)^k\cdot n^{O(1)}$ for any constant $\delta>0$. For graphs with 
polynomial edge weights the running time was improved to $2^k\cdot n^{O(1)}$ by 
Nederlof~\cite{DBLP:journals/algorithmica/Nederlof13} using the technique of 
fast subset convolution. \pname{Steiner Forest} is the generalization where the 
input contains an edge-weighted graph $G$ and a list $(s_1,t_1), \dots, 
(s_d,t_d)$ of pairs of terminals and the task is to find a minimum-cost subgraph 
containing an $s_i$--$t_i$ path for every $1\le i \le d$. The observation that 
the connected components of the solution to \pname{Steiner Forest} induces a 
partition on the set $R=\{s_1,\dots,s_d,t_1,\dots,t_d\}$ of terminals such that 
each class of the partition forms a tree, implies the fixed-parameter 
tractability of \pname{Steiner Forest} parameterized by $k=|R|$: we can solve 
the problem by for example trying every partition of $R$ and invoking a 
\pname{Steiner Tree} algorithm for each class of the partition.

On directed graphs, Steiner problems can become significantly harder, and while 
there is a richer landscape of variants, only few results are known 
\cite{DBLP:journals/siamdm/GuoNS11,DBLP:journals/talg/ChekuriEGS11,
DBLP:journals/siamcomp/FeldmanR06,DBLP:journals/jal/CharikarCCDGGL99,
DBLP:journals/algorithmica/Zelikovsky97,DBLP:journals/siamcomp/ChitnisFHM20,
DBLP:conf/iwpec/ChitnisEHKKS14,DBLP:conf/esa/ChitnisFM18, 
DBLP:conf/stacs/EibenKPS19}.  A natural and well-studied generalization of 
\pname{Steiner Tree} to directed graphs is \pname{Directed Steiner Tree (DST)}, 
where an arc-weighted directed graph $G$ and terminals $r,t_1,\dots, t_d$ are 
given and the task is to find a minimum-cost subgraph containing an $r\to t_i$ 
path for every $1\le i \le d$. Using essentially the same techniques as in the 
undirected 
case~\cite{DBLP:journals/algorithmica/Nederlof13,fuchs2007dynamic,
DBLP:journals/networks/DreyfusW71}, one can show that this problem is also FPT 
parameterized by the number of terminals $k=d+1$. An equally natural 
generalization of \pname{Steiner Tree} to directed graphs is the \pname{Strongly 
Connected Steiner Subgraph (SCSS)} problem, where an arc-weighted directed graph 
$G$ with terminals $t_1, \dots, t_k$ is given, and the task is to find a 
minimum-cost subgraph containing a $t_i\to t_j$ path for any $1\le i,j\le k$ 
with $i\neq j$. Guo et al.~\cite{DBLP:journals/siamdm/GuoNS11} showed that, 
unlike \pname{DST}, the \pname{SCSS} problem is W[1]-hard parameterized by $k$ 
(see also~\cite{DBLP:journals/siamcomp/ChitnisFHM20}), and is thus unlikely to 
be FPT for this parameter (for more background on parameterized complexity 
theory see~\cite{flum2006parameterized}). A common generalization of \pname{DST} 
and \pname{SCSS} is the \pname{Directed Steiner Network (DSN)} problem (also 
called \pname{Directed Steiner Forest}\footnote{Note however that unlike 
\pname{Steiner Forest}, the solution to \pname{DSN} is not necessarily a forest, 
which justifies the use of the alternative name used here.} or 
\pname{Point-to-Point Connection}), where an arc-weighted directed graph $G$ and 
a list $(s_1,t_1), \dots, (s_d,t_d)$ of terminal pairs are given and the task is 
to find a minimum-cost subgraph containing an $s_i\to t_i$ path for every $1\le 
i \le d$. Being a generalization of \pname{SCSS}, the \pname{Directed Steiner 
Network} problem is also W[1]-hard for the number of terminals $k$ in the set 
$R=\{s_1,\dots,s_d,t_1,\dots,t_d\}$, but Feldman\footnote{We note that Jon 
Feldman (co-author of~\cite{DBLP:journals/siamcomp/FeldmanR06}) is not the same 
person as Andreas Emil Feldmann (co-author of this paper).} and 
Ruhl~\cite{DBLP:journals/siamcomp/FeldmanR06} showed that the problem is 
solvable in time $n^{O(d)}$, that is, in polynomial time for every constant 
$d=O(k^2)$.

\newcommand{\HDSN}{\pname{$\mc{H}$-DSN}}
Besides \pname{Directed Steiner Tree}, what other special cases of 
\pname{Directed Steiner Network} are fixed-parameter tractable? Our main result 
gives a complete map of the complexity landscape of directed Steiner problems 
on general input graphs, 
precisely describing all the \mbox{FPT/W[1]-hard} variants and revealing highly 
non-trivial generalizations of \pname{Directed Steiner Tree} that are still 
tractable. Our results are expressed in the following formal framework.  The 
pairs $(s_1,t_1), \dots, (s_d,t_d)$ in the input of \pname{DSN} can be 
interpreted as a directed (unweighted) \emph{pattern graph} on a set $R$ of 
terminals. If this pattern graph is an out-star, then the problem is precisely 
\pname{DST}; if it is a bidirected clique, then the problem is precisely 
\pname{SCSS}. More generally, if $\mc{H}$ is any class of graphs, then we define 
the \pname{Directed Steiner $\mc{H}$-Network ($\mc{H}$-DSN)} problem as the 
restriction of \pname{DSN} where the pattern graph is a member of~$\mc{H}$. 
That is, the input of \HDSN\ is an arc-weighted directed graph $G$, a set 
$R\subseteq V(G)$ of terminals, and an unweighted directed graph $H\in \mc{H}$ 
on $R$; the task is to find a minimum-cost subgraph $N\subseteq G$ 
(``network'') such that $N$ contains an $s\to t$ path for every $st\in E(H)$.

We give a complete characterization of the classes $\mc{H}$ for which \HDSN\ is 
FPT or \mbox{W[1]-hard}. We need the following definition of 
``almost-caterpillar graphs'' to describe the borderline between the easy and 
hard cases (see \autoref{fig:caterpillar}).

\begin{figure}[t]
\centering
\includegraphics[scale=1.5]{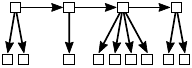}
\hspace{2cm}
\includegraphics[scale=1.5]{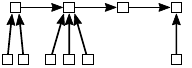}
\caption{Two \mbox{$4$-cater}\-pillars: an out-cater\-pillar (on the left) and 
an in-cater\-pillar (on the right).}
\label{fig:caterpillar}
\end{figure}

\begin{dfn}\label{dfn:caterpillar}
A \emph{$\lambda_0$-caterpillar} graph is constructed as follows. Take a 
directed path $(v_1,\ldots,v_{\lambda_0})$ from $v_1$ to $v_{\lambda_0}$, and 
let $W_1,\ldots,W_{\lambda_0}$ be pairwise disjoint vertex sets such that 
$v_i\in W_i$ for each $i\in\{1,\ldots,\lambda_0\}$. Now add edges such that 
either every $W_i$ forms an out-star with root $v_i$, or every $W_i$ forms an 
in-star with root~$v_i$. In the former case we also refer to the resulting 
$\lambda_0$-caterpillar as an \emph{out-caterpillar}, and in the latter as an 
\emph{in-caterpillar}. A $0$-caterpillar is the empty graph. The class 
$\mc{C}_{\lambda,\delta}$ contains all directed graphs $H$ such that there is a 
set of edges $F\subseteq E(H)$ of size at most $\delta$ for which the remaining 
edges $E(H)\setminus F$ span a $\lambda_0$-caterpillar for some 
$\lambda_0\leq\lambda$.
\end{dfn}

If there is an $s\to t$ path in the pattern graph $H$ for two terminals
$s,t\in R$, then adding the edge $st$ to $H$ does not change the
problem: connectivity from $s$ to $t$ is already implied by $H$, hence
adding this edge does not change the feasible solutions. That is,
adding a transitive edge does not change the solution space and hence
it is really only the transitive closure of the pattern $H$ that
matters. We say that two pattern graphs are \emph{transitively equivalent} if 
their transitive closures are isomorphic. We denote the class of patterns that 
are transitively equivalent to some pattern of $\mc{C}_{\lambda,\delta}$ by 
$\mc{C}^*_{\lambda,\delta}$. Our main result is a sharp dichotomy saying that 
\HDSN\ is FPT if every pattern of $\mc{H}$ is transitively equivalent to an 
almost-caterpillar graph and it is W[1]-hard otherwise. In order to provide 
reductions for the hardness results we need the technical condition that the 
class of patterns is \emph{recursively enumerable}, i.e., there is some 
algorithm, which enumerates all members of the class. In the FPT cases, we make 
the algorithmic result more precise by stating a running time that is expressed 
as a function of $\lambda$, $\delta$, and the \emph{vertex cover number} $\tau$ 
of the input pattern $H$, i.e., $\tau$ is the size of the smallest vertex subset 
$W$ of $H$ such that every edge of $H$ is incident to a vertex of $W$.

\begin{thm}\label{thm:main}
Let $\mc{H}$ be a recursively enumerable class of patterns.
\begin{enumerate}
\item If there are constants $\lambda$ and $\delta$ such that 
$\mc{H}\subseteq\mc{C}^*_{\lambda,\delta}$, then \HDSN\ with parameter $k=|R|$ 
is FPT and can be solved in $2^{O(k+\tau\omega\log\omega)}n^{O(\omega)}$ time, 
where $\omega=(1+\lambda)(\lambda+\delta)$ and $\tau$ is the vertex cover number 
of the given input pattern~$H\in\mc{H}$.
\item Otherwise, if there are no such constants $\lambda$ and $\delta$, 
then the problem is \mbox{\textup{W[1]}-hard} for parameter~$k$.
\end{enumerate}
\end{thm}

In \autoref{thm:main}(1), the reason for the slightly complicated runtime is 
that the algorithm was optimized to match the runtime of some previous 
algorithms in special cases. In particular, invoking \autoref{thm:main} with 
specific classes $\mc{H}$, we can obtain algorithmic or hardness results for 
specific problems. For example, we may easily recover the following facts:
\begin{itemize}
\item If $\mc{H}_\textup{DST}$ is the class of all out-stars, then 
$\mc{H}_\textup{DST}$-DSN is precisely the \pname{DST} problem. As 
$\mc{H}_\textup{DST}\subseteq \mc{C}^*_{1,0}$ holds, \autoref{thm:main}(1) 
recovers the fact that \pname{DST} can be solved in time 
$2^{O(k)}n^{O(1)}$ and is hence FPT parameterized by the number $k$ of 
terminals~\cite{DBLP:journals/algorithmica/Nederlof13,fuchs2007dynamic,
DBLP:journals/networks/DreyfusW71}.
\item If $\mc{H}_\textup{SCSS}$ is the class of all bidirected cliques (or 
equivalently the class of all directed cycles), then 
$\mc{H}_\textup{SCSS}$-DSN is precisely the \pname{SCSS} problem. One can 
observe that $\mc{H}_\textup{SCSS}$ is not contained in 
$\mc{C}^*_{\lambda,\delta}$ for any constants $\lambda,\delta$ (for example, 
because every graph in $\mc{C}_{\lambda,\delta}$ has at most $\lambda+2\delta$ 
vertices with both positive in-degree and positive out-degree, and this remains 
true also for the graphs in $\mc{C}^*_{\lambda,\delta}$). Hence 
\autoref{thm:main}(2) recovers the fact that \pname{SCSS} is 
W[1]-hard~\cite{DBLP:journals/siamdm/GuoNS11}. Note that any pattern of 
$\mc{H}_\textup{SCSS}$ is transitively equivalent to a bidirected star with less 
than $2k$ edges, so that $\mc{H}_\textup{SCSS}\subseteq \mc{C}^*_{0,2k}$. Since 
a star has vertex cover number $\tau=1$, for \pname{SCSS} our algorithm in 
\autoref{thm:main}(1) recovers the running time of $2^{O(k\log 
k)}n^{O(k)}=n^{O(k)}$ given by Feldman and 
Ruhl~\cite{DBLP:journals/siamcomp/FeldmanR06}. We note however, that the 
constants in the degree of the polynomial are larger in our case compared 
to~\cite{DBLP:journals/siamcomp/FeldmanR06}.
\item Let $\mc{H}_d$ be the class of directed graphs with at most $d$
  edges. As $\mc{H}_d\subseteq \mc{C}^*_{0,d}$ holds,
  \autoref{thm:main}(1) recovers the fact that \pname{Directed Steiner 
Network} with at most $d$ demands is polynomial-time solvable for every 
constant $d$~\cite{DBLP:journals/siamcomp/FeldmanR06}.
\item Recently, Such\'y \cite{suchy-steiner} studied the following 
generalization of \pname{DST} and \pname{SCSS}: in the \pname{$q$-Root Steiner 
Tree ($q$-RST)} problem, a set of $q$ roots and a set of leaves are given, 
and the task is to find a minimum-cost network where the roots are in the same 
strongly connected component and every leaf can be reached from every root. 
Building on the work of Feldman and Ruhl~\cite{DBLP:journals/siamcomp/FeldmanR06},
Such\'y~\cite{suchy-steiner} presented an algorithm with running time 
$2^{O(k)}\cdot n^{O(q)}$ for this problem, which shows that it is FPT for 
every constant~$q$. Let $\mc{H}_{\textup{$q$-RST}}$ be the class of directed 
graphs that are obtained from an out-star by making~$q-1$ of the edges 
bidirected. Observe that $\mc{H}_{\textup{$q$-RST}}$ is a subset 
of~$\mc{C}_{1,q-1}$, that \pname{$q$-RST} can be expressed by an instance of 
$\mc{H}_{\textup{$q$-RST}}$-\pname{DSN}, and that any pattern of 
$\mc{H}_{\textup{$q$-RST}}$ has vertex cover number $\tau=1$. Thus 
\autoref{thm:main}(1) implies that \pname{$q$-RST} can be solved in time 
$2^{O(k+q\log q)}\cdot n^{O(q)}=2^{O(k)}\cdot n^{O(q)}$, recovering the fact 
that it is FPT for every constant~$q$.
\end{itemize}

Thus the algorithmic side of \autoref{thm:main} unifies and
 generalizes three algorithmic results: the fixed-parameter tractability
 of \pname{DST} (which is based on dynamic programming on the tree structure of 
the solution), \pname{$q$-RST} (which is based on simulating a ``pebble game''), 
but also the polynomial-time solvability of \pname{DSN} with constant number of 
demands (which also is based on simulating a ``pebble game''). 
 Let us point out that our algorithmic results are significantly more general
 than just the unification of these three results: the generalization
 from stars to bounded-length caterpillars is already a significant
 extension and very different from earlier results. We consider it a
 major success of the systematic investigation that, besides finding
 the unifying algorithmic ideas generalizing all previous
 results, we were able to find tractable special cases in an
 unexpected new direction.

 There is a surprising non-monotonicity in the classification result of
 \autoref{thm:main}. As \pname{DST} is FPT and \pname{SCSS} is
 W[1]-hard, one could perhaps expect that \HDSN\ becomes harder
 as the pattern become denser. However, it is possible that the
 addition of further demands makes the problem easier. For example,
 if $\mc{H}$ contains every graph that is the vertex-disjoint union of
 two out-stars, then \HDSN\ is classified to be W[1]-hard by
 \autoref{thm:main}(2). However, if we consider those graphs where
 there is also a directed edge from the center of one star to the
 other star, then these graphs are \mbox{2-cater}\-pillars (i.e., contained in
 $\mc{C}_{2,0}$) and hence \HDSN\ becomes FPT by
 \autoref{thm:main}(1). This unexpected non-monotonicity further
 underlines the importance of completely mapping the complexity
 landscape of the problem area: without complete classification, it
 would be very hard to predict what other tractable/intractable
 special cases exist.

We mention that one can also study the vertex-weighted version of the problem, 
where the input graph has weights on the vertices and the goal is to minimize 
the total vertex-weight of the solution. In general, vertex-weighted problems 
can be more challenging than edge-weighted variants 
\cite{DBLP:journals/talg/DemaineHK14,DBLP:conf/icalp/BateniHL13,
DBLP:journals/jal/KleinR95,DBLP:conf/soda/ChekuriHKS07}. However, for general 
directed graphs, there are easy transformations between the two variants. Thus 
the results of this paper can be interpreted for the vertex-weighted version as 
well. 

\subsection{Our techniques}

We prove \autoref{thm:main} the following way. In \autoref{sec:cw}, we first 
establish the combinatorial bound that there is a solution whose cutwidth, and 
hence also (undirected) treewidth,\footnote{Throughout this paper we use only 
the undirected treewidth, as formally defined in \autoref{sec:prelim}.} is 
bounded by the number of demands.

\begin{thm}\label{lem:cutwidth}
A minimal solution $M$ to a pattern $H$ has cutwidth at most $7d$ if $d=|E(H)|$.
\end{thm}

This serves as the first step, which we exploit in \autoref{sec:tw} to prove 
that if the pattern is an almost-caterpillar in~$\mc{C}^*_{\lambda,\delta}$, 
then the (undirected) treewidth of the optimum solution can be bounded by a 
function of $\lambda$ and~$\delta$.

\begin{thm}\label{crl:tw}
The treewidth of a minimal solution to any pattern graph in 
$\mc{C}^*_{\lambda,\delta}$ is at most $7(1+\lambda)(\lambda+\delta)$.
\end{thm}

To prove the above two theorems we thoroughly analyze the combinatorial 
structure of minimal solutions, by untangling the intricate interplay between 
the $s\to t$ paths in a given solution for demands $st$ of a pattern graph $H$. 
The resulting bounds can then be exploited in an algorithm that restricts the 
search for a bounded-treewidth solution (\autoref{sec:alg}). To obtain this 
algorithm we generalize dynamic programming techniques for other settings to the 
\pname{DSN} case, by introducing novel tools to tackle the intricacies of this 
problem.

\begin{thm}\label{crl:alg}
Let an instance of $\mc{H}$-DSN be given by a graph $G$ with $n$ vertices, and 
a pattern $H$ on $k$ terminals with vertex cover number $\tau$. If the 
optimum solution to $H$ in $G$ has treewidth~$\omega$, then the optimum can be 
computed in $2^{O(k+\tau\omega\log\omega)}n^{O(\omega)}$ time.
\end{thm}

Combining \autoref{crl:tw} and \autoref{crl:alg} proves the algorithmic side of 
\autoref{thm:main}. We remark that the proof is completely self-contained (with 
the exception of some basic facts on treewidth) and in particular we do not 
build on the algorithms of Feldman and 
Ruhl~\cite{DBLP:journals/siamcomp/FeldmanR06}. As combining 
\autoref{lem:cutwidth} and \autoref{crl:alg} already proves that \pname{DSN} 
with a constant number of demands can be solved in polynomial time, as a 
by-product this gives an independent proof for the result of Feldman and 
Ruhl~\cite{DBLP:journals/siamcomp/FeldmanR06}. One can argue which algorithm is 
simpler, but perhaps our proof (with a clean split of a combinatorial and an 
algorithmic statement) is more methodological and better reveals the underlying 
reason why the problem is tractable.

Finally, in \autoref{sec:hard} we show that whenever the patterns in $\mc{H}$ 
are not transitively equivalent to almost-caterpillars, the problem is 
W[1]-hard. Our proof follows a novel, non-standard route. We first show that 
there is only a small number of obstacles for not being transitively equivalent 
to almost-caterpillars: the graph class contains (possibly after identification 
of vertices) arbitrarily large strongly connected graphs, pure diamonds, or 
flawed diamonds (see \autoref{lem:only-if} for the precise statement). Showing 
the existence of these obstacles needs a non-trivial combinatorial argument. We 
then provide a separate W[1]-hardness proof for each of these obstacles, 
completing the proof of the hardness side of \autoref{thm:main}.

\subsection{Subsequent related work}

Since the publication of the conference 
version~\cite{DBLP:conf/icalp/FeldmannM16} of this paper several results have 
appeared that build on our work. Especially the algorithm of \autoref{crl:alg} 
has been used as a subroutine to solve several special cases of \pname{DSN}. We 
survey some of these results here.

\paragraph{Parameterizing by the number $k$ of terminals.}
As mentioned above, the algorithm for \pname{DSN} based on simulating a ``pebble 
game'' by Feldman and Ruhl~\cite{DBLP:journals/siamcomp/FeldmanR06} has a faster 
runtime of~$n^{O(d)}$ than implied by \autoref{crl:alg}, where $d$ is the 
number of demands. Measured in the stronger parameter $k$ (which can be smaller 
than $d$ up to a quadratic factor) the Feldman and 
Ruhl~\cite{DBLP:journals/siamcomp/FeldmanR06} algorithm runs in~$n^{O(k^2)}$ 
time. Interestingly, Eiben et al.~\cite{DBLP:conf/stacs/EibenKPS19} show that 
this is essentially best possible, as no $f(k)n^{o(k^2/\log k)}$ time algorithm 
exists for \pname{DSN} for any computable function~$f$, under the Exponential 
Time Hypothesis (ETH). However, as summarized below, in special cases it is 
possible to beat this lower bound. 

\paragraph{Planar and bounded genus graphs.}
A directed graph is considered planar if its underlying undirected graph is. For 
such inputs Chitnis et al.~\cite{DBLP:journals/siamcomp/ChitnisFHM20} show that 
under ETH no $f(k)n^{o(k)}$ time algorithm can solve~\pname{DSN}. Eiben et 
al.~\cite{DBLP:conf/stacs/EibenKPS19} show that an optimum solution of genus~$g$ 
has treewidth $2^{O(g)}k$ and thus \autoref{crl:alg} implies an algorithm with 
runtime $2^{O(k^2\log k)}n^{O(k)}$ for graphs of constant genus, matching the 
previous runtime lower bound for planar graphs. However, for the special case of 
the \pname{SCSS} problem, Chitnis et 
al.~\cite{DBLP:journals/siamcomp/ChitnisFHM20} prove that in planar graphs there 
exists a faster algorithm with runtime~$2^{O(k)}n^{O(\sqrt{k})}$. To obtain such 
an algorithm, in the conference version 
of~\cite{DBLP:journals/siamcomp/ChitnisFHM20} the authors devise a 
generalization of the ``pebble games'' of Feldman and 
Ruhl~\cite{DBLP:journals/siamcomp/FeldmanR06} for \pname{SCSS} in planar graphs. 
However,  in the journal version~\cite{DBLP:journals/siamcomp/ChitnisFHM20} the 
authors use \autoref{crl:alg} to get a much cleaner and simpler proof, by 
showing that any optimum solution has treewidth~$O(\sqrt{k})$. They also obtain 
a matching runtime lower bound of $f(k)n^{o(\sqrt{k})}$ for \pname{SCSS} on 
planar graphs.

\paragraph{Bidirected graphs.} An interesting application of \autoref{crl:alg} 
is the \pname{SCSS} problem on bidirected graphs, i.e., directed graphs for 
which an edge~$uv$ exists if and only if its reverse edge~$vu$ also exists and 
has the same weight. While this problem remains NP-hard, Chitnis et 
al.~\cite{DBLP:conf/esa/ChitnisFM18} show that it is FPT parameterized by $k$, 
which is in contrast to general input graphs where the problem is W[1]-hard (as 
also implied by \autoref{thm:main}). To show this result, it is not enough to 
bound the treewidth of a solution and then apply \autoref{crl:alg} directly, as 
is done for the above mentioned problems on planar graphs. In fact, there are 
examples~\cite{DBLP:conf/esa/ChitnisFM18} in which the optimum solution to 
\pname{SCSS} on bidirected graphs has treewidth $\Theta(k)$. Nevertheless, as 
shown in~\cite{DBLP:conf/esa/ChitnisFM18} it is possible to decompose the 
optimum solution to this problem into \emph{poly-trees}, i.e., directed graphs 
of (undirected) treewidth~$1$. As a consequence, an FPT algorithm can guess the 
decomposition of the optimum, and apply \autoref{crl:alg} repeatedly 
(with~$\omega=1$) to compute all poly-tree solutions. This algorithm can be made 
to run in $2^{k^2+O(k)}n^{O(1)}$ time. In contrast, for the more general 
\pname{DSN} problem on bidirected graphs, Chitnis et 
al.~\cite{DBLP:conf/esa/ChitnisFM18} show that no $f(k)n^{o(k/\log k)}$ time 
algorithm exists, under ETH.

\paragraph{Planar bidirected graphs.}
If the input graph is both planar and bidirected, then Chitnis et 
al.~\cite{DBLP:conf/esa/ChitnisFM18} show that the treewidth of any optimum 
solution to \pname{DSN} is $O(\sqrt{k})$. \autoref{crl:alg} then implies an 
algorithm with runtime $2^{O(k^{3/2}\log k)}n^{O(\sqrt{k})}$, which is faster 
than possible in planar graphs but also in bidirected graphs, as mentioned 
above. Furthermore, they show that \autoref{crl:alg} can be used to obtain a 
parameterized approximation scheme for \pname{DSN} on planar bidirected graphs, 
i.e., a $(1+\eps)$-approximation can be computed in 
$2^{O(k^2)}n^{2^{O(1/\eps)}}$ time. For this they prove that the optimum 
solution to \pname{DSN} in planar bidirected graphs can be covered by a set of 
\pname{DSN} solutions, each of which only contains $2^{O(1/\eps)}$ terminals, 
and such that the sum of the costs of all these solutions is only a 
$(1+\eps)$-fraction more than the optimum. Similar to \pname{SCSS} on bidirected 
graphs, the idea now is to guess how these solutions cover the terminal set, and 
then compute all of them using the above mentioned $2^{O(k^{3/2}\log 
k)}n^{O(\sqrt{k})}$ time algorithm for \pname{DSN} on planar bidirected graphs 
(which follows from \autoref{crl:alg}). Since each of the solutions only 
contains $2^{O(1/\eps)}$ terminals, the degree of the polynomial depends only on 
$\eps$ every time the algorithm of \autoref{crl:alg} is executed.

\subsection{Preliminaries}\label{sec:prelim}

In this paper, we are mainly concerned with directed graphs, i.e., graphs for 
which every edge is an ordered pair of vertices. For convenience, we will also 
give definitions, such as the treewidth, for directed graphs, even if they are 
usually defined for undirected graphs. For any graph $G$ we denote its vertex 
set by $V(G)$ and its edge set by $E(G)$. We denote a directed edge from $u$ to 
$v$ by~$uv$, so that $u$ is its \emph{tail} and $v$ is its \emph{head}. We say 
that both $u$ and $v$ are \emph{incident} to the edge~$uv$, and $u$ and $v$ are 
\emph{adjacent} if the edge $uv$ or the edge $vu$ exists. We refer to $u$ and $v$ as the 
\emph{endpoints} of $uv$. For a vertex $v$ the \emph{in-degree (out-degree)} is 
the number of edges that have $v$ as their head (tail). A \emph{source (target)} 
is a vertex of in-degree $0$ (out-degree~$0$). An \emph{in-arborescence 
(out-arborescence)} is a connected graph with a unique target (source), also 
called its \emph{root}, such that every vertex except the root has 
out-degree~$1$ (in-degree~$1$). The \emph{leaves} of an in-arborescence 
(out-arborescence) are its sources (targets). A \emph{$u\to v$ path} is an 
out-arborescence with root $u$ and a single leaf~$v$, and its \emph{length} is 
its number of edges. A \emph{star} $S$ with \emph{root} $u$ is a graph in which 
every edge is incident to $u$. All vertices in a star different from the root 
are called its \emph{leaves}. An \emph{in-star (out-star)} is a star which is an 
in-arborescence (out-arborescence).
A \emph{strongly connected component} (SCC) $H$ of a directed graph $G$ is an 
inclusion-wise maximal sub-graph of $G$ for which there is both a $u\to v$ path and 
a $v\to u$ path for every pair of vertices $u,v\in V(H)$. A \emph{directed 
acyclic graph} (DAG) is a graph in which every SCC is a singleton, i.e., it 
contains no cycles.

The following observation is implicit in previous work 
(cf.~\cite{DBLP:journals/siamcomp/FeldmanR06}) and will be used throughout this 
paper. Here we consider a \emph{minimal} solution $M$ to an instance of DSN, in 
which no edge can be removed without making the solution infeasible.

\begin{lem}\label{lem:arb}
Consider an instance of DSN where the pattern $H$ is an out-star (resp., 
in-star) with root $t\in R$. Then any minimal solution $M$ to $H$ is an 
out-arborescence (resp., in-arborescence) rooted at $t$ for which every leaf is 
a terminal.
\end{lem}
\begin{proof}
We only prove the case when $H$ is an out-star, as the other case follows by 
symmetry. Suppose for contradiction that $M$ is not an out-arborescence. As it 
is clear that $M$ is connected and $t$ is the unique source in a minimal solution, $M$ not being an 
out-arborescence implies that there is a vertex $v\in V(M)$ with in-degree at 
least $2$, i.e., there are two distinct edges $e$ and $f$ of $M$ that have $v$ 
as their head. Since $M$ is a minimal solution, removing $e$ disconnects some 
terminal $\ell$ from $t$, which in particular means that there is a $t\to \ell$ 
path $P$ going through~$e$. Clearly, this path cannot go through $f$, as both 
$e$ and $f$ have the same head $v$. Thus if we remove $f$, then any terminal 
$\ell'$ will remain being reachable from $t$: we may reroute any $t\to \ell$ 
path $Q$ that passed through $f$ via a path through $e$ instead by following $P$ 
from $t$ to~$v$, the head of~$f$, and then following $Q$ from $v$ to $\ell'$. 
This however contradicts the minimality of~$M$.
\end{proof}

A \emph{tree decomposition} $D$ of a graph $G$ is an undirected tree for which 
every node $w\in V(D)$ is associated with a set $b_w\subseteq V(G)$ called a 
\emph{bag}. Additionally it satisfies the following properties:
\begin{enumerate}[(a)]
\item\label{item:tw-edges} for every edge $uv\in E(G)$ there is a bag $b_w$  
for some $w\in V(D)$ containing it, i.e., $u,v\in b_w$, and 
\item\label{item:tw-vertices} for every vertex $v\in V(G)$ the nodes of $D$ 
associated with the bags containing $v$ induce a non-empty and connected 
subgraph of~$D$.
\end{enumerate}
The \emph{width} of the tree decomposition is $\max\{|b_w|-1\mid w\in V(D)\}$. 
The \emph{treewidth} of a graph $G$ is the minimum width of any tree 
decomposition for~$G$. It is known (by an easy folklore proof) that for any 
graph $G$ of treewidth $\omega$ there is a \emph{smooth} tree decomposition $D$ 
of $G$, which means that $|b_w|=\omega+1$ for all nodes $w$ of $D$ and $|b_w\cap 
b_{w'}|=\omega$ for all adjacent nodes $w,w'$ of $D$.

\section{The cutwidth of minimal solutions for bounded-size patterns}
\label{sec:cw}

\newcommand{\Vicut}{(V_i,\overline{V}_i)}

The goal of this section is to prove \autoref{lem:cutwidth}: we bound the 
\emph{cutwidth} of a minimal solution~$M$ to a pattern $H$ in terms of 
$d=|E(H)|$. A \emph{layout} of a graph $G$ is an injective function 
$\psi:V(G)\rightarrow\mathbb{N}$ inducing a total order on the vertices of $G$. 
Given a layout, we define the set $V_i=\{v\in V(G)\mid \psi(v)\leq i\}$ and say 
that an edge \emph{crosses the cut $\Vicut$} if it has one endpoint in $V_i$ and 
one endpoint in $\overline{V}_i:=V(G)\setminus V_i$. The \emph{cutwidth of the 
layout} is the maximum number of edges crossing any cut $\Vicut$ for any 
$i\in\mathbb{N}$. The cutwidth of a graph is the minimum cutwidth over all its 
layouts.

  Like Feldman and Ruhl~\cite{DBLP:journals/siamcomp/FeldmanR06}, we consider 
the two extreme cases of \emph{directed acyclic graphs (DAGs)} and 
\emph{strongly connected components (SCCs)} in our proof. Contracting all SCCs 
of a minimal solution $M$ without removing parallel edges sharing the same head 
and tail, but removing the resulting self-loops, produces a directed acyclic 
multi-graph~$D$, the so-called \emph{condensation graph} of $M$. We bound the 
cutwidth of $D$ and the SCCs of $M$ separately, and then put together these two 
bounds to obtain a bound for the cutwidth of~$M$. As we will see, bounding the 
cutwidth of the acyclic multi-graph $D$ and putting together the bounds are 
fairly simple. The main technical part is bounding the cutwidth of the SCCs.

We will need two simple facts about cutwidth. First, the cutwidth of an acyclic multi-graph can be bounded using the 
existence of a \emph{topological ordering} of the vertices. That is, 
for any acyclic graph $G$ there is an injective function 
$\varphi:V(G)\rightarrow\mathbb{N}$ such that $\varphi(u)<\varphi(v)$ if $uv\in 
E(G)$. Note that such a function in particular is a layout. 

\begin{lem}\label{lem:cw-of-dag}
If $D$ is an acyclic directed 
multi-graph~$D$ that is the union of $d$ paths and $\varphi_D$ is an arbitrary 
topological ordering of $D$, then the layout given by $\varphi_D$ has cutwidth 
at most $d$.
\end{lem}
\begin{proof}%
To bound the cutwidth, we argue that a path $P$ crosses any cut $\Vicut$ at 
most once. Note that no edge can have a head $v$ 
and tail $u$ with $\varphi_D(v)\leq\varphi_D(u)$, since $\varphi_D$ is a 
topological ordering. In particular, for the first edge $uv$ of $P$ crossing 
$\Vicut$ we get $\varphi_D(u)\leq i<\varphi_D(v)$. For any vertex $w$ reachable 
from $v$ on the path, the transitivity of the topological order implies 
$i<\varphi_D(w)$ so that $w$ cannot be the tail of an edge crossing the cut. 
Thus no second edge of the path $P$ crosses $\Vicut$. As $D$ is the union of 
$d$ paths, each cut is crossed by at most $d$ edges of~$D$. 
\end{proof}

The next lemma shows that bounding the cutwidth of each SCC and the condensation 
graph of $G$, bounds the cutwidth of $G$.

\begin{lem}\label{lem:cw-sum}
Let $G$ be a directed graph and $D$ be its condensation multi-graph. If the 
cutwidth 
of $D$ is $x$ and the cutwidth of every SCC of $G$ is at most $y$, then the 
cutwidth of $G$ is at most~$x+y$.
\end{lem}
\begin{proof}%
Let SCC$(u)\subseteq G$ be the SCC of $G$ that was contracted into the vertex 
$u$ in $D$. If a vertex $u$ of $G$ was not contracted, then SCC$(u)$ is the 
singleton~$u$. For each $u\in V(D)$, there exists a layout $\varphi_u$ of 
SCC$(u)$ with cutwidth at most $y$, while for $D$ there exists a layout 
$\varphi_D$ with cutwidth at most $x$. Let $\mu=\max\{\varphi_u(v)\mid u\in 
V(D)\land v\in$~SCC$(u)\}$ be the maximum value taken by any layout of an SCC. 
We define a layout $\psi$ of $G$ as $\psi(v)=\mu\cdot\varphi_D(u)+\varphi_u(v)$, 
where $v\in$~SCC$(u)$. Since the topological orderings are injective, $\psi$ is 
injective, and the intervals $[\mu\cdot\varphi_D(u)+1,\mu\cdot\varphi_D(u)+\mu]$ 
of values that $\psi$ can take for vertices of different SCCs are disjoint. 
Hence for any $i\in\mathbb{N}$, there is at most one SCC of $G$ whose edges 
cross the cut $\Vicut$, and so the cutwidth of $\psi$ is at most the cutwidth of 
any $\varphi_u$ plus the cutwidth of~$\varphi_D$. 
\end{proof}

\noindent Let us now bound the cutwidth of the SCCs.

\begin{lem}\label{lem:SCC-bound}
Any SCC $U$ of a minimal solution $M$ to a pattern $H$ with at most $d$ edges 
has cutwidth at most $6d$.
\end{lem}
\begin{proof}
First we establish that $U$ is a minimal solution to a certain pattern.
\begin{claim}\label{lem:SCC-pattern}
$U$ is a minimal solution to a pattern $H_U$ with at most $d$ edges.
\end{claim}
\begin{proof}%
  Consider a path $P_{st}$ in $M$ from $s$ to $t$ for some edge
  $st\in E(H)$. Let $v$ be the first vertex of $U$ on the path
  $P_{st}$, and let $w$ be the last.  Note that all vertices on
  $P_{st}$ between $v$ and $w$ must be contained in $U$ since
  otherwise $U$ would not be an SCC. Hence we can construct a pattern
  graph $H_U$ for $U$ with an edge $vw$ for the first and last vertex
  of each such path $P_{st}$ in $M$ that contains vertices of $U$. The
  SCC must be a minimal solution to the resulting pattern since a
  superfluous edge would also be removable from the minimal solution $M$: any 
edge $e$ of
  $U$ needed in $M$ by some edge $st\in E(H)$ also has a corresponding
  edge $vw$ in the pattern $H_U$ that needs it, i.e., all paths from
  $v$ to $w$ in $U$ pass through $e$. Since $H_U$ has at most one edge
  for each path $P_{st}$ in $M$ with $st\in E(H)$, the pattern $H_U$ has at 
most $d=|E(H)|$ edges. \cqed
\end{proof}

Let $R_U$ be the terminals in the pattern $H_U$ given by
\autoref{lem:SCC-pattern} and let us select an arbitrary root
$t\in R_U$. Note that $H_U$ has at most $d$ edges and hence
$|R_U|\le 2d$. Let $S_{in}$ (resp., $S_{out}$) be an in-star (resp.,
out-star) connecting $t$ with every other vertex of~$R_U$. As $U$ is a
strongly connected graph containing every vertex of $R_U$, it is also
a solution to the pattern $S_{in}$ on~$R_U$. Let us select an
$A_{in}\subseteq U$ that is a minimal solution to $S_{in}$; by 
\autoref{lem:arb}, $A_{in}$ is an in-arborescence with at most $2d$
leaves. Similarly, let $A_{out}\subseteq U$ be an out-arborescence
that is a minimal solution to $S_{out}$. Observe that $U$ has to be
exactly $A_{in}\cup A_{out}$: if there is an edge $e\in E(U)$ that is
not in $A_{in}\cup A_{out}$, then $U\setminus e$ still contains a path
from every vertex of $R_U$ to every other vertex of $R_U$ through $t$,
contradicting the fact that $U$ is a minimal solution to pattern
$H_U$.

Let $Z$ be the set of edges obtained by reversing the edges in 
$E(A_{in})\setminus E(A_{out})$.  As reversing edges does not change
the cutwidth, bounding the cutwidth of $A_{out}\cup Z$ will also imply a bound
on the cutwidth of $U=A_{in}\cup A_{out}$.

\begin{claim}\label{lem:dag}
  The union $A_{out}\cup Z$ is a directed acyclic graph.
\end{claim}
\begin{proof}%
Assume that $A_{out}\cup Z$ has a cycle $O$. We will identify a superfluous edge 
in $U$, which contradicts its minimality. Note that $Z$ is a forest of 
out-arborescences, and thus $O$ must contain edges from both $A_{out}$ and $Z$. 
Among the vertices of $O$ that are incident to edges of the in-arborescence 
$A_{in}$, pick one that is closest to the root $t$ in~$A_{in}$. Let $P$ be the 
path from this vertex $v$ to $t$ in~$A_{in}$. From $v$ we follow the edges of 
the cycle $O$ in their reverse direction, to find a path $Q\subseteq O\cap 
A_{out}$ of maximal length leading to $v$ and consisting of edges not in $Z$. 
Let $u$ be the first vertex of~$Q$ (where possibly $u=v$). The edge $wu$ on $O$ 
that has $u$ as its head must be an edge of~$Z$, since $Q$ is of maximal length. 
Note also that this edge exists since $O$ contains edges from both $A_{out}$ and 
$Z$.

Now consider the in-arborescence $A_{in}$, which contains the reverse edge 
$uw\in E(A_{in})\setminus E(A_{out})$ and the path $P$ from $v$ to $t$. Since 
$v$ is a closest vertex from $O$ to $t$ in $A_{in}$, the path $P$ cannot 
contain~$uw$ (otherwise $w$ would be closer to $t$ than $v$). However, this 
means 
that removing $uw$ from $M$ will still leave a solution to~$H$: any path 
connecting through $uw$ to $t$ can be rerouted through $Q$ and then $P$, while 
no connection from $t$ to a terminal needed $uw$ as it is not in $A_{out}$. 
Hence for every edge in the pattern $H$, there is still a path connecting the 
respective terminals through $t$. Thus $U$ was not minimal, which is a 
contradiction.
\cqed
\end{proof}

\autoref{lem:dag} implies a topological ordering on the vertices of 
$A_{out}\cup Z$. This order can be used as a layout for $U$. Using some more 
structural insights, the number of edges crossing a given cut can be bounded by a function of 
the number of edges of the pattern graph, as the following claim shows. 

\begin{claim}\label{lem:cw-of-scc}
Any topological ordering $\varphi$ of the graph 
$A_{out}\cup Z$ has cutwidth at most~$6d$.
\end{claim}
\begin{proof}%
To bound the number of edges crossing a cut given by the layout $\varphi$, we 
will consider edges of $A_{out}$ and $Z$ separately, starting with the former. 
Obviously $\varphi$ also implies a topological ordering of the 
subgraph~$A_{out}$. As the out-arborescence $A_{out}$ has at most $2d$ leaves, 
it is the union of at most $2d$ paths, each starting in $t$ and ending at a 
terminal. By \autoref{lem:cw-of-dag}, the cutwidth of $\varphi$ for edges of 
$A_{out}$ is at most~$2d$.

Recall that $V_i=\{v\in V(G)\mid \psi(v)\leq i\}$.
To bound the number of edges of $Z$ crossing a cut~$\Vicut$, recall that $uv\in 
Z$ if and only if the reverse edge $vu$ is in $E(A_{in})\setminus E(A_{out})$. 
Consider the set $B=E(A_{out})\cap E(A_{in})$ of edges that are shared by both 
arborescences. These are the only edges that are not reversed in $A_{in}$ to 
give $A_{out}\cup Z$. 
Let $B^*$ consist of the edges of $B$ that cross the cut $\Vicut$. As 
$B\subseteq E(A_{out})$ and the cutwidth of $\varphi$ for the edges of $A_{out}$ 
is at most~$2d$, we have that $|B^*|\le 2d$. Consider the graph obtained by 
removing $B^*$ from $A_{in}$, so that $A_{in}$ falls into a forest of 
in-arborescences. Each leaf of this forest is either a leaf of $A_{in}$ or 
incident to the head of an edge of $B^*$. Since $A_{in}$ has at most $2d$ leaves 
and $|B^*|\le 2d$, the number of leaves of the forest is at most $4d$. This 
means that the forest is the union of at most $4d$ paths, each starting in a 
leaf and ending in a root of an in-arborescence. Let $\mc{P}$ denote the set of 
all these paths. 

Consider a path $P$ of $\mc{P}$, which is a directed path of $A_{in}$. We show 
that $P$ can cross the cut $\Vicut$ at most once. Recall that every edge of $P$ 
is either an edge of $A_{out}$ or an edge of~$Z$ reversed. Whenever an edge of 
$P$ crosses the cut $\Vicut$, then it has to be an edge of $Z$ reversed: 
otherwise, it would be an edge of $E(A_{in})\cap E(A_{out})$, and such edges are 
in $B^*$, which cannot be in $P$ by definition. Thus if $uv$ is an edge of $P$ 
crossing $\Vicut$, then $vu\in Z$, and the topological ordering $\varphi$ 
implies that $\varphi(v)\leq i <\varphi(u)$. In other words, every edge of $P$ 
is crossing the cut from the right to the left, so clearly at most one such edge 
can be in $P$. This gives an upper bound of $|\mc{P}|\le 4d$ on the number of 
edges of $Z$ crossing the cut, completing the required $6d$ upper bound. 
\cqed\end{proof}

As the underlying undirected graph of $U$ and $A_{out}\cup Z$ are the same, 
\autoref{lem:cw-of-scc} implies that the cutwidth of $U$ is at most $6d$. This 
completes the proof of \autoref{lem:SCC-bound}.
\end{proof}

The proof of \autoref{lem:cutwidth} follows easily from putting together the 
ingredients.
\begin{proof}[Proof (of \autoref{lem:cutwidth})]
Consider a minimal solution $M$ and let $D$ be its condensation graph. The 
minimum solution $M$ is the union of $d$ directed paths and this is true also 
for the contracted condensation graph~$D$. Hence \autoref{lem:cw-of-dag} shows 
that $D$ has cutwidth at most $d$. By \autoref{lem:SCC-bound}, each SCC of $M$ 
has cutwidth at most $6d$. Thus \autoref{lem:cw-sum} implies that the cutwidth 
of $M$ is at most $7d$.
\end{proof}

We remark that the bound on the cutwidth in \autoref{lem:cutwidth} is 
tight up to a constant factor: 
Take a constant degree expander on $d$ vertices. It has treewidth 
$\Omega(d)$~\cite{DBLP:journals/jct/GroheM09}, and so its cutwidth is at least 
as large. Now bi-direct 
each (undirected) edge $\{u,v\}$ by replacing it with the directed edges $uv$ 
and~$vu$. Next subdivide every edge $uv$ to obtain edges $ut$ and $tv$ for a 
new vertex~$t$, and make $t$ a terminal of $R$. This yields a strongly 
connected instance $G$. The pattern graph $H$ for this instance is a cycle 
on~$R$, which has $\Theta(d)$ edges, since the terminals are subdivision points 
of bi-directed edges of a constant degree graph with $d$ vertices. As $H$ is 
strongly connected, every minimal solution to $H$ contains the edges $ut$ and 
$tv$ incident to each terminal $t$. Thus a minimal solution contains all of $G$ 
and has cutwidth $\Omega(d)$. 

\section{The treewidth of minimal solutions to almost-caterpillar patterns} 
\label{sec:tw}

In this section, we prove that any minimal solution $M$ to a pattern 
$H\in\mc{C}^*_{\lambda,\delta}$ has the following structure.

\begin{thm}\label{thm:structure}
A minimal solution $M$ to a pattern $H\in\mc{C}^*_{\lambda,\delta}$ is the union of
\begin{itemize}
\item a subgraph $M^c$ (``core'') that is a minimal solution to a sub-pattern 
$H^c$ of $H$, where the latter has at most $(1+\lambda)(\lambda+\delta)$ edges, 
and 
\item a forest $M-E(M^c)$ of either out- or in-arborescences, each of which 
intersects~$M^c$ only at its root.
\end{itemize}
\end{thm}

According to \autoref{lem:cutwidth}, the cutwidth of the \emph{core} $M^c$ is 
therefore at most $7(1+\lambda)(\lambda+\delta)$. It is 
well known~\cite{DBLP:journals/eatcs/Bodlaender88} that the cutwidth is 
an upper bound on the treewidth of a graph, and so also the treewidth of $M^c$ 
is at most $7(1+\lambda)(\lambda+\delta)$. It is easy to see that attaching any 
number of arborescences to $M^c$ does not increase the treewidth. Thus 
we obtain \autoref{crl:tw}, which is the basis for our algorithm to solve 
$\mc{H}$-DSN in case every pattern of $\mc{H}$ is transitively equivalent to an 
almost-caterpillar.

In particular, when adding $\delta$ edges to the pattern of the DST problem, 
which is a single out-star, i.e., a $1$-caterpillar, then the pattern becomes a 
member of $\mc{C}_{1,\delta}$ and hence our result implies a linear treewidth 
bound of $O(\delta)$. The example given at the end of \autoref{sec:cw} also 
shows that there are patterns $H\in\mc{C}_{\lambda,\delta}$ for which every 
minimal solution has treewidth $\Omega(\lambda+\delta)$: just consider the case 
when $H$ is a cycle of length~$\lambda+\delta$ (i.e., it contains a trivial 
caterpillar graph). One interesting question is whether the treewidth bound of 
$7(1+\lambda)(\lambda+\delta)$ in \autoref{crl:tw} is tight. We conjecture that 
the treewidth of any minimal solution to a pattern graph 
$H\in\mc{C}^*_{\lambda,\delta}$ is actually $O(\lambda+\delta)$.

\begin{proof}[Proof (of \autoref{thm:structure})]
  Let $M$ be a minimal solution to a pattern
  $H\in\mc{C}^*_{\lambda,\delta}$. Since every pattern in
  $\mc{C}^*_{\lambda,\delta}$ has a transitively equivalent pattern in
  $\mc{C}_{\lambda,\delta}$ and replacing a pattern with a
  transitively equivalent pattern does not change the space of
  feasible solutions, we may assume that $H$ is actually in
  $\mc{C}_{\lambda,\delta}$, i.e., $H$ consists of a caterpillar of
  length at most $\lambda$ and $\delta$ additional edges.

The statement is trivial if $|E(H)|\leq \delta \le (1+\lambda)(\lambda+\delta)$. Otherwise, according to 
\autoref{dfn:caterpillar}, $H$ contains a $\lambda_0$-caterpillar for some $1\le 
\lambda_0\leq\lambda$ and at most $\delta$ additional edges. Hence let us fix a 
set $F$ of at most $\delta$ edges of $H$ such that the remaining edges of $H$ 
form a $\lambda_0$-caterpillar $C$ for some $1\le \lambda_0\leq\lambda$ with a 
path $(v_1,\ldots,v_{\lambda_0})$ on the roots of the stars~$S_i$. We only 
consider the case when $C$ is an out-caterpillar as the other case is symmetric, 
i.e., every $S_i$ is an out-star. Define $I$ to be the subgraph of $H$ spanned 
by all edges of $H$ except the edges of the stars, i.e., 
$E(I)=E(H)\setminus\bigcup_{i=1}^{\lambda_0} E(S_i)$. Note that $|E(I)|\leq 
\lambda_0+\delta$. We fix a subgraph $M_I$ of $M$ that is a minimal solution to 
the sub-pattern~$I$, and for every $st\in E(I)$ we fix a path $P_{st}$ in~$M_I$. 
Note that $M_I$ is the union of these at most $\lambda+\delta$ paths, since 
$M_I$ is a minimal solution. For each star $S_i$, let us consider a minimal 
solution $M_{S_i}\subseteq M$ to $S_i$; note that $M_{S_i}$ has to be an 
out-arborescence by \autoref{lem:arb}.

For some $i\in\{1,\ldots,\lambda_0\}$, let $\ell$ be a leaf of $S_i$, and let 
$e$ be an edge of $M$. If $M\setminus e$ has no path from $v_i$ to $\ell$, then 
we say that $e$ is \emph{$\ell$-necessary}. More generally, we say that $e$ is 
\emph{$i$-necessary} if $e$ is $\ell$-necessary for some leaf $\ell$ of $S_i$.

\begin{claim}\label{lem:witness}
Let $P$ be a path in $M$, and for some $i\in\{1,\ldots,\lambda_0\}$ let 
$W_i\subseteq E(M)$ contain all $i$-necessary edges $f$ for which $f\notin 
E(P)$, but the head of $f$ is a vertex of $P$. Then there exists one leaf $\ell$ 
of $S_i$ such that every $f\in W_i$ is $\ell$-necessary.
\end{claim}
\begin{proof}
Since all edges of $W_i$ are contained in the out-arborescence $M_{S_i}$, no two 
of them have the same head. Hence we can identify the first edge $e\in W_i$ for 
the path~$P$, i.e., the edge for which the head of every other edge in $W_i$ 
can 
be reached from $e$'s head on $P$. Since $e$ is $i$-necessary, it is $\ell$-necessary for some leaf $\ell$ of $S_i$.  We claim that every other edge of $W_i$ is also $\ell$-necessary. Assume 
the opposite, which means that there is a path $Q$ in $M$ from $v_i$ to $\ell$ 
that does not contain some~$f\in W_i$. On the other hand, every path 
(including~$Q$) from $v_i$ to $\ell$ in $M$ contains the $\ell$-necessary 
edge~$e$. This means that there is a path from $v_i$ through $e$ and $P$ that 
reaches 
the head of~$f$, and this path does not pass through $f$. Hence for any path 
that goes from $v_i$ to some leaf of $S_i$ via $f$, there is an 
alternative route that avoids $f$. This however contradicts the fact that $f$ 
is $i$-necessary.\cqed
\end{proof}

Using this observation, we identify the core $M^c$ of $M$ using the at most 
$\lambda+\delta$ paths $P_{st}$ that make up~$M_I$, and then selecting an 
additional at most $\lambda_0$ paths for each~$P_{st}$, one for each star of 
the 
caterpillar. To construct $M^c$ together with its pattern graph~$H^c$, we 
initially let $M^c=M_I$ and $H^c=I$ and repeat the following step for every 
$st\in E(I)$ and $1\le i \le \lambda_0$. For a given $st$ and $i$, let us check 
if there are $i$-necessary edges $f\notin E(P_{st})$ that have their heads on 
the path $P_{st}\subseteq M_I$. If so, then by \autoref{lem:witness} all these 
edges are $\ell$-necessary for some leaf $\ell$ of $S_i$. We add an arbitrary 
path of $M$ from $v_i$ to $\ell$ (which contains all these edges) to $M^c$ and 
add the edge $v_i\ell$ to $H^c$. After repeating this step for every $st\in 
E(H)$ and $i$, we remove superfluous edges from~$M^c$: as long as there is an 
edge $e\in E(M^c)$, which can be removed while maintaining feasibility for the 
pattern $H^c$, i.e., for every $vw\in E(H^c)$ there is a $v\to w$ path in $M^c$ 
not containing $e$, we remove $e$. Finally, we remove any isolated vertices from 
$M^c$.

Note that the resulting network $M^c$ is a minimal solution to $H^c$ by construction. Also note 
that $H^c$ contains at most $\lambda+\delta$ edges from $I$ and at most 
$\lambda_0\leq\lambda$ additional edges for each edge of $I$, so that 
$|E(H^c)|\leq (1+\lambda)(\lambda+\delta)$. We prove that the remaining graph 
$M^c- E(M)$ consists of out-arborescences, each of which intersects $M^c$ only 
at the root. For this, we rely on the following key observation.

\begin{claim}\label{lem:same-head}
If a vertex $u$ has at least two incoming edges in $M$, then every such edge is 
in the core $M^c$.
\end{claim}
\begin{proof}
  First we show that there is an $st\in E(I)$ such that every $s\to t$ path in 
$M$ goes through~$u$. Suppose for contradiction that for every $st\in E(I)$
  there is a path from $s$ to $t$ in $M$ avoiding~$u$. Since $M$ is a minimal 
solution, the edges entering $u$ must then be needed
  for some stars $S_i$ of the pattern $H$ instead. Let $e$ and $f$ be two edges entering $u$. As $e$ and $f$ have
  the same head, they cannot be part of the same out-arborescence
  $M_{S_i}$. Therefore, there are indices $i<j$ such that (w.l.o.g.)
  $e$ is $i$-necessary and $f$ is $j$-necessary.

There is a path in $M$ from the root $v_i$ of $S_i$ to the root $v_j$ of $S_j$, 
due to the path $(v_1,\ldots,v_{\lambda_0})$ in the caterpillar $C\subseteq H$. 
Since path $(v_1,\ldots,v_{\lambda_0})$ is part of $I$, our assumption on
$e$ and $f$ implies that there is a path $P$ in $M$ from $v_i$ to $v_j$ that avoids both $e$ and $f$. As $f\in E(M_{S_j})$, there is a path $Q$ in $M$ starting in $v_j$ 
and passing through $f$. This path cannot contain $e$, as $e$ and $f$ 
have the same head~$u$. The existence of $P$ and $Q$ implies that $u$ can be 
reached from $v_i$ by a path through $v_j$ and~$f$, avoiding the edge $e$. 
Thus for any edge $v_i \ell\in E(S_i)$, if there is a $v_i\to \ell$ path going 
through $e$ (and hence vertex $u$), then it can be rerouted to avoid $e$ and use 
edge $f$ instead. This however contradicts the fact that $e$ is $i$-necessary.

We now know that there is an $st\in E(I)$ such that every $s\to t$
path in $M$ goes through~$u$. Suppose that there is an edge
$e\not\in E(M^c)$ entering $u$.  If $e$ is needed for some
$s't'\in E(I)$ in~$M$, then $e$ is also present in $M^c$, and we are
done. Otherwise, as $M$ is a minimal solution, edge $e$ is
$i$-necessary for some $i\in\{1,\ldots,\lambda_0\}$. Consider now the
step in the construction of $M^c$ when we considered $st\in E(I)$ and
integer $i$. As we have shown, the $s\to t$ path $P_{st}$ goes through
$u$. Thus $e$ is an $i$-necessary edge not in $E(P_{st})$ such that
its head is on $P_{st}$. This means that we identified a leaf $\ell$
of $S_i$ such that $e$ is $\ell$-necessary, introduced $v_i\ell$ into
$H^c$, and added a $v_i\to \ell$ path to $H^c$, which had to contain
$e$. Moreover, since all paths from $v_i$ to $\ell$ in $M$ pass
through $e$, edge $e$ then remains in $M^c$ when removing superfluous
edges.  \cqed
\end{proof}

We are now ready to show that every component of the remaining part is an 
out-arborescence and intersects the core only at the root. 
\begin{claim}\label{lem:M-M^c}
The remaining graph $M^+:=M- E(M^c)$ is a forest of out-arbores\-cences, each 
of which intersects~$M^c$ only at the root.
\end{claim}
\begin{proof}%
If $M^+$ is not a forest of out-arborescences, then there must be two edges in 
$M^+$ with the same head or there must be a directed cycle in $M^+$. The former 
is excluded by \autoref{lem:same-head}. For the latter, first note that if an 
edge $e\in E(M)$ is not $i$-necessary for any $i\in\{1,\ldots,\lambda_0\}$, then 
it is needed for $I$, since $M$ is a minimal solution. Hence $e$ was added to 
$M^c$ as a part of $M_I$, and remained in $M^c$ even after removing superfluous 
edges, as $E(I)\subseteq E(H^c)$. In particular, this means that every edge of 
$M^+$ is part of some~$M_{S_i}$. Furthermore, any directed cycle $O$ in $M^+$ 
must contain edges from at least two out-arborescences $M_{S_i}$ and $M_{S_j}$ 
with~$i<j$. If one of the roots $v_i$ or $v_j$ of $M_{S_i}$ and~$M_{S_j}$, 
respectively, is not part of $O$, then there is a path from $v_i$ or $v_j$ 
leading to $O$. In case both $v_i$ and $v_j$ are part of $O$, we also get such a 
path, since $M_I$ contains a path from $v_i$ to $v_j$, but $M_I$ contains no 
edges of $O$. Hence there must be a vertex $u$ on $O$ that is the head of two 
edges of which one belongs to $O\subseteq M^+$. However this is again excluded 
by \autoref{lem:same-head}, and so $M^+$ contains no cycle.

For the second part of the claim, assume that an out-arborescence of $M^+$ 
intersects $M^c$ at a vertex $u$ that is not its root. As noted above, any edge 
that is not $i$-necessary for any $i\in\{1,\ldots,\lambda_0\}$ is part of the 
core $M^c$. Hence there is an edge $e\in E(M^+)$ that has $u$ as its head and is 
$i$-necessary for some $i\in\{1,\ldots,\lambda_0\}$. There must be at least one 
edge of $M^c$ incident to~$u$, since $u\in V(M^c)$ and we removed all isolated 
vertices from $M^c$. The in-degree of $u$ is 0 in~$M^c$, since 
\autoref{lem:same-head} and $e\notin E(M^c)$ implies that the in-degree of $u$ 
is~exactly $1$ in $M$. Because $M^c$ is a minimal solution to~$H^c$ and $u$ has 
in-degree $0$ in $M^c$, there is at least one edge of $H^c$ whose tail is $u$: 
the (at least 1) edges going out from $u$ can be used only by paths starting 
at~$u$. Suppose first that there is an edge $uw\in E(H^c)$ and that it is from 
$E(I)$. Consider the step of the construction of $M^c$ and $H^c$ when we 
considered the edge $uw$ and the integer $i$. The path $P_{uw}$ is starting at 
$u$, and edge $e$ is an $i$-necessary edge with $e\not\in E(P_{uw})$ whose head 
is on $P_{uw}$. Thus we have identified a leaf $\ell$ of $S_i$ such that $e$ is 
$\ell$-necessary, introduced $v_i\ell$ into~$H^c$, and added a $v_i\to \ell$ 
path to $H^c$, which had to contain $e$. As $e$ is $\ell$-necessary, it would 
have remained in $M^c$ even after removing superfluous edges, contradicting 
$e\not\in E(M^c)$. Thus we can conclude that there is no edge of $I$ with $u$ as 
its tail.  This means that if $uw\in E(H^c)$, then it is only possible that $u$ 
is the root $v_{\lambda_0}$ of the last star $S_{\lambda_0}$, as every other 
root $v_j$ with $j<\lambda_0$ is incident to the edge $v_jv_{j+1}$ of~$I$. 
Moreover, if $\lambda_0>1$, then $v_{\lambda_0-1}v_{\lambda_0}\in E(I)$, which 
leads to a contradiction, since then $M^c$ would contain a path from 
$v_{\lambda_0-1}$ to $u=v_{\lambda_0}$, but the only edge entering $u$ is $e$ 
and we have $e\notin E(M^c)$. Thus $i=\lambda_0=1$ is the only possibility. This 
however would mean that the arborescence $M_{S_i}$ contains a cycle, as $e\in 
E(M_{S_i})$ and the head of $e$ is the root $v_i$ of $M_{S_i}$. This leads to a 
contradiction, and so we can conclude that no out-arborescence of $M^+$ 
intersects $M^c$ at a vertex different from its root.  \cqed
\end{proof}

Since we have already established that $M^c$ is a minimal solution to 
$H^c$ with
$|E(H^c)|\leq (1+\lambda)(\lambda+\delta)$, \autoref{lem:M-M^c} completes the proof  of \autoref{thm:structure}.
\end{proof}

\section{An algorithm to find optimal solutions of bounded treewidth: proof of 
 \autoref{crl:alg}}
\label{sec:alg}

This section is devoted to proving \autoref{crl:alg}.
That is, we present an algorithm based on dynamic programming that computes the
optimum solution to a given pattern $H$, given that the treewidth of
the optimum is bounded by~$\omega$, and given that the vertex cover
number of~$H$ is $\tau$.  Roughly speaking, we will exploit the first
property by guessing the bags of the tree decomposition of the optimum
solution, which can be done in $n^{O(\omega)}$ time as the size of a
bag is at most~$\omega+1$. Since each bag forms a separator of the
optimum, we are able to precompute partial solutions connecting a
subset of the terminals to a separator. We need to also guess the
subset of the terminals for which there are $2^k$ choices. These
partial solutions are then put together at the separators to form
larger partial solutions containing more vertices. The algorithm
presented in this section is not the most obvious one: it was
optimized to exploit that the vertex cover number of $H$ is
$\tau$. While this optimization requires the implementation of
additional ideas and makes the algorithm more complicated, it allows
us to replace a factor of $2^{O(k\omega)}$ in the running time with the
potentially much smaller~$2^{O(k+\tau\omega\log\omega)}$ 
(i.e.,~if~$\tau=o(k/\log\omega)$).

\paragraph{Defining the dynamic programming table.}
Our algorithm maintains a table~$T$, where in each entry we aim at storing a 
partial solution of minimum cost that provides partial connectivity of certain 
type between the terminals contained in the network and a separator $U$ of the 
solution. The entries are computed by recursively putting together partial 
solutions. In order to do this, we also need to keep track of how vertices of a 
separator~$U$ of a partial solution are connected to each other. For this we 
need the following formal definitions encoding the internal connectivity of $U$ 
and the types of connectivity between terminals and $U$.

A minimal solution $M\subseteq G$ (and therefore also any optimum solution) to a 
pattern $H$ is the union of $d=|E(H)|$ paths~$P_{st}$, one for each edge $st\in 
E(H)$. Throughout this section, given a minimal solution~$M$ we fix such a path 
$P_{st}$ for each demand $st$, and we let $\mc{P}$ denote the set of all these 
paths. Let now $N\subseteq M$ be a partial solution of a minimal solution $M$. 
For a separator $U\subseteq V(G)$ the \emph{$U$-projection} of $N$ encodes the 
connectivity that $N$ provides between the vertices of $U$ by short-cutting each 
path $P_{st}\in\mc{P}$ to its restriction on $U$. Formally, it is a set of edges 
containing $uv\in U^2$ if and only if there is an edge $st\in E(H)$ and a $u\to 
v$ subpath $P$ of some $P_{st}\in\mc{P}$ such that $P$ is contained in $N$ and 
the internal vertices of $P$ do not belong to~$U$. Note that the path $P$ can 
also be an edge $uv\in E(N[U])$ induced by $U$ in~$N$, and the $U$-projection 
will in fact contain any such edge, since the edge must be part of some path 
$P_{st}\in\mc{P}$ of the minimal solution~$M$. On the other hand, observe that 
even if $N$ contains a $u\to v$ path with internal vertices not in~$U$, it is 
not necessarily true that $uv$ is in the $U$-projection: we put $uv$ into the 
$U$-projection only if there is a path in $\mc{P}$ that has a $u\to v$ subpath. 
Thus this definition is more restrictive than just expressing the connectivity 
provided by $N$, as it encodes only the connectivity essential for the paths in 
$\mc{P}$.
The exact significance of this subtle difference will be apparent later: for 
example, this more restrictive definition implies fewer edges in the 
$U$-projection, which makes the total number of possibilities smaller.

The property that $H$ has vertex cover number $\tau$ implies that $H$ is the 
union of $c\leq 2\tau$ in- and out-stars $S_1,\ldots,S_c$. We denote the root of 
$S_j$ by $r_j$ and its leaf set by~$L_j$. Let also $R_{in}$ and $R_{out}$ 
contain the roots $r_j$ of all in- and out-stars, respectively. By 
\autoref{lem:arb}, a minimal solution $M$ to $H$ is the union of $c$ 
arborescences, each with at most $|V(M)|-1$ edges. Note that any path 
$P_{st}\in\mc{P}$ implies the existence of a set of edges in a $U$-projection of 
$M$ forming a path in the $U$-projection. It is not difficult to see that if we take any arborescence that is 
the union of paths in $\mc{P}$, then its $U$-projection is a forest of arborescences with 
at most $|U|-1$ edges in the $U$-projection: for example, in case of an out-arborescence, it is not possible that two distinct edges enter the same vertex of $U$ in the projection. Therefore, if we have
$|U|\leq\omega+1$, then the fact that $\mc{P}$ is the union of $c$ arborescences implies 
that the $U$-projection of every partial solution $N\subseteq M$ contains at 
most $c\omega$ edges. 
Note that here it becomes essential how we defined the $U$-projection: even if 
$N$ consists of only a single path $P_{st}$ going through every vertex of $U$, 
it is possible that there are $\binom{|U|}{2}$ pairs $uv\in U^2$ such that 
$N$ has a $u\to v$ path; however, with our definition, only $|U|-1$ edges would 
appear in the $U$-projection. 

We now describe the {\em type} of connectivity provided by a partial solution 
$N$ by a tuple $(Q,I,B,\mc{A})$, which is defined in the following way.  First, 
$Q$ is the set of terminals appearing in $N$, and $I$ is the subgraph of the 
partial solution induced by $U$, i.e., $I=N[U]$. The set $B\subseteq (U\times 
R_{in})\cup(R_{out}\times U)\cup(U\times U)$ describes how $N$ provides 
connectivity between the vertices of the separator~$U$, and between the 
separator $U$ and the roots, as follows. First, an edge $uv\in U\times U$ 
appears in~$B$ if $uv$ is in the $U$-projection of $N$. Moreover, an edge $uv\in 
(U\times R_{in})\cup(R_{out}\times U)$ is in $B$ if there is a path $P_{st}\in 
\mc{P}$ that has a $u\to v$ subpath in~$N$ (regardless of what internal vertices 
this subpath has).

The last item $\mc{A}$ requires more explanation. Consider an out-star $S_j$ 
rooted at $r_j\in R_{out}$ outside of $N$ and let $\ell\in L_j$ be one of its 
leaves for which $\ell\in V(N)$. Intuitively, to classify the type of 
connectivity provided by $N$ to the leaf $\ell$, we should describe the subset 
$U_\ell\subseteq U$ of vertices from which $\ell$ is reachable in $N$. Then we 
know that in order to extend $N$ into a full solution where $\ell$ is reachable 
from $r_j$, we need to ensure that an $r_j\to v$ path exists for some~$v\in 
U_\ell$. However, describing these sets $U_\ell$ for every leaf $\ell\in L_j$ 
may result in an unacceptably high number of different types (of 
order~$2^{O(k\omega)}$), which we cannot afford to handle in the claimed running 
time. Thus we handle the leaves in a different way. For every root $r_j\in 
R_{out}$, we define a set $A_{j}\subseteq U$ as follows. Initially we set 
$A_j=\emptyset$ and then we consider every leaf $\ell\in L_j\cap V(N)$ one by 
one. Let $P$ be the $r_j\to \ell$ path in $\mc{P}$. If $P\subseteq N$, then 
there is nothing to be done for this leaf $\ell$ and we can proceed with the 
next leaf. Otherwise, suppose that the {\em maximal suffix} of $P$ in $N$ starts 
at vertex~$w$, that is, $w$ is the first vertex of $P$ such that the $w\to \ell$ 
subpath of $P$ is a subgraph of $N$. If $w\not\in U$, then we declare the type 
of $N$ as {\em invalid.} Otherwise, we extend $A_j$ with $w$ and proceed with 
the next leaf in $L_j\cap V(N)$. This way, we define a set $A_j$ for every root 
$r_j\in R_{out}$ and in a similar manner, we can define a set $A_j$ for every 
root $r_j\in R_{in}$ as well (then $w$ is defined to be the last vertex of the 
$\ell\to r_j$ path $P$ such that the $\ell\to w$ subpath is in $N$). The family 
$\mc{A}=(A_1,\dots,A_c)$ in the tuple $(Q,I,B,\mc{A})$ is the collection of 
these sets $A_j$.

The table $T$ used in the dynamic programming algorithm has entries 
$T[i,Q,U,I,B,\mc{A}]$, where $i\le |V(G)|$ is an integer, $Q\subseteq R$ is a 
subset of terminals, $U$ is a subset of at most $\omega+1$ vertices, $I$ is a 
subgraph of $G[U]$ with at most $c\omega$ edges, $B$ is a subset of $ (U\times 
R_{in})\cup(R_{out}\times U)\cup(U\times U)$ with $|B\cap(U\times U)|\le 
c\omega$, and $\mc{A}=(A_1,\dots,A_c)$ with every $A_j$ being a subset of 
$U$. 
We say that a network $N\subseteq G$ \emph{satisfies} such an entry if the 
following properties hold:
\begin{enumerate}[(P1)]
 \item $N$ has at most $i$ vertices, which include $U$ and has $V(N)\cap 
R=Q$,\label{it:vert}
 \item $I$ is the graph induced by $U$ in $N$, i.e., $I=N[U]$,\label{it:I}
 \item for every edge $uv\in B$ there is a $u\to v$ path in $N$, 
and\label{it:B}
 \item for any out-star (resp., in-star) $S_j$ with root $r_j$ and any $\ell\in 
L_j\cap V(N)$, there is a $w\to \ell$ path (resp., $\ell\to w$ path) in $N$ for 
some $w\in A_j\cup\{r_j\}$. \label{it:R}
\end{enumerate}

Let $N\subseteq M$ be an induced subgraph of the minimal solution $M$. We say 
that $N$ is {\em $U$-attached} in $M$ if $U\subseteq V(N)$ and the neighbourhood 
of each vertex in $V(N)\setminus U$ is fully contained in $V(N)$. The following 
statement is straightforward from the definition.
\begin{lem}\label{prop:satself}
  If $N\subseteq M$ is a $U$-attached induced subgraph of $M$ with $i$
  vertices, then it has a valid type $(Q,I,B,\mc{A})$ and $N$
  satisfies the entry $T[i,Q,U,I,B,\mc{A}]$.
\end{lem}

\paragraph{The algorithm.}
For each entry $T[i,Q,U,I,B,\mc{A}]$ the following simple algorithm computes 
some network satisfying properties~\ref{it:vert} to~\ref{it:R}, for increasing 
values of~$i$. It first computes entries for which $i\leq\omega+1$ by simply 
checking whether the graph $I$ satisfies properties~\ref{it:vert} 
to~\ref{it:R}. If it does then $I$ is stored in the entry, and otherwise the 
entry remains empty.
For values $i>\omega+1$, the entries are computed 
recursively by combining precomputed networks with a smaller number of vertices. 
The algorithm sets the entry $T[i,Q,U,I,B,\mc{A}]$ to the minimum cost network 
$N$ that has properties~\ref{it:vert} to~\ref{it:R} and for which 
$N=T[i_1,Q_1,U_1,I_1,B_1,\mc{A}_1]\cup T[i_2,Q_2,U_2,I_2,B_2,\mc{A}_2]$ for some 
$i_1,i_2<i$. Again, if no such network exists, we leave the entry empty. 

\paragraph{Correctness of the algorithm.}
According to this algorithm any non-empty entry of the table stores some 
network that has properties~\ref{it:vert} to~\ref{it:R}.
The following lemma shows that for certain entries of the table our algorithm 
computes an optimum partial solution. 
Recall from \autoref{sec:prelim} that we may assume that a given tree 
decomposition is smooth, i.e., if the treewidth is $\omega$ then 
$|b_w|=\omega+1$ and $|b_w\cap b_{w'}|=\omega$ for any two adjacent nodes 
$w,w'$ of the decomposition tree.
If $D'$ is a subtree of a tree decomposition $D$ and $b_w$ a bag of~$D'$, we say 
that \emph{$D'$ is attached via $b_w$ in $D$} if $w$ is the only node of $D'$ 
adjacent to nodes of $D$ not in $D'$.

\begin{lem}\label{lem:alg-tw}
Let $D_M$ be a smooth tree decomposition of $M$, where the treewidth of $M$ 
is~$\omega$. Let~$D$ be a subtree of $D_M$ attached via a bag $U$ in $D_M$ and 
let $N\subseteq M$ be the sub-network of~$M$ induced by all vertices contained 
in the bags of $D$. Then $N$ has a valid type $(Q,I,B,\mc{A})$ and satisfies the 
entry $T[i,Q,U,I,B,\mc{A}]$ for $i=|V(N)|$.  Moreover, at the end of the 
algorithm the entry $T[i,Q,U,I,B,\mc{A}]$ contains a network satisfying the 
entry and with cost at most that of~$N$.
\end{lem}
\begin{proof}
The first statement follows from \autoref{prop:satself}.
The proof of the second statement is by induction on the number of nodes in tree 
decomposition $D$ of $N$. If $D$ contains only one node, which is associated 
with the bag $U$, then the statement is trivial, since in this case 
$i=|V(N)|=|U|=\omega+1$ and $N=N[U]=I$ by \autoref{prop:satself}, so that the 
algorithm stores $N$ in the entry. If $D$ contains at least two nodes, let $w_1$ 
be the node corresponding to $U$, and let~$w_2$ be an adjacent node to~$w_1$ in 
$D$. The edge $w_1w_2$ separates the tree $D$ into two subtrees. 
For~$h\in\{1,2\}$, let $D_h$ be the corresponding subtree of $D$ 
containing~$w_h$, i.e., their disjoint union is $D$ minus the edge~$w_1w_2$. If 
$N_1$ and $N_2$ are the sub-networks of $N$ induced by the bags of $D_1$ and 
$D_2$, respectively, then $N=N_1\cup N_2$. Let $U_h$ be the set of vertices in 
the bag corresponding to the node $w_h$ (in particular $U_1=U$) and let 
$i_h=|V(N_h)|$. It is easy to see that $N_h$ is a $U_h$-attached induced 
subgraph of $M$ for $h=1,2$. Hence \autoref{prop:satself} implies that $N_h$ has 
a valid type $(Q_h,I_h,B_h,\mc{A}_h)$ with $\mc{A}_h=\{A_{1}^h,\ldots,A_{c}^h\}$ 
and satisfies the entry $T[i_h,Q_h,U_h,I_h,B_h,\mc{A}_h]$ for~$h\in\{1,2\}$. As 
$D$ is attached via $U$ in $D_M$, clearly also $D_h$ is attached via $U_h$ in 
$D_M$. Moreover, since $D$ is smooth we have $U_1\setminus U_2\neq\emptyset$ and 
$U_2\setminus U_1\neq\emptyset$. In particular, for each~$h\in\{1,2\}$, there is 
some vertex $v$ of $N$, which is not contained in~$N_h$, as the bags containing 
$v$ must form a connected subtree of $D$. Thus $i_h<i$ and $i>|U|=\omega+1$. 
Furthermore, by induction we may assume that the entry 
$T[i_h,Q_h,U_h,I_h,B_h,\mc{A}_h]$ contains a network $N'_h$ satisfying the entry 
and with cost at most that of $N_h$. Thus by the following claim, the union of 
the two networks $N'_1$ and $N'_2$ stored in the entries 
$T[i_h,Q_h,U_h,I_h,B_h,\mc{A}_h]$ for $h\in\{1,2\}$, respectively, is considered 
by the algorithm as a candidate to store in the entry $T[i,Q,U,I,B,\mc{A}]$.

\begin{claim}\label{clm:union}
The solutions $N'_1$ and $N'_2$ that are stored in the respective entries 
$T[i_1,Q_1,U_1,I_1,B_1,\mc{A}_1]$ and $T[i_2,Q_2,U_2,I_1,B_1,\mc{A}_2]$ can be 
combined to a solution $N'=N'_1\cup N'_2$ satisfying $T[i,Q,U,I,B,\mc{A}]$.
\end{claim}
\begin{proof}
By induction, $N'_1$ and $N'_2$ have property~\ref{it:vert}, and so the 
solutions $N'_1$ and $N'_2$ have at most $i_1$ and $i_2$ vertices, respectively, 
and $N'_1$ contains $U_1$ and $Q_1$, while $U_2$ and $Q_2$ are contained in 
$N'_2$. Since $U_1\cap U_2$ separates $N_1$ and~$N_2$, we have that $V(N_1)\cap 
V(N_2)=U_1\cap U_2$, and so~$i=i_1+i_2-|U_1\cap U_2|$, as $i=|V(N)|$ and
$i_h=|V(N_h)|$ for $i\in\{1,2\}$. Hence the union~$N'$ of $N'_1$ and $N'_2$ 
contains $U=U_1$ and $Q=Q_1\cup Q_2$, and has at most $i$ vertices, so that we 
obtain property~\ref{it:vert} for~$N'$. For property~\ref{it:I}, by induction 
$N'_1$ has property~\ref{it:I}, so that $I_1=N'_1[U_1]$. We also have 
$I=N[U]=N_1[U_1]=I_1$, since $N$ has property~\ref{it:I}, $U=U_1$, and by 
definition of $N_1$ and~$I_1$. Hence $I=N'_1[U_1]=N'[U]$, and so we obtain 
property~\ref{it:I} for~$N'$.

For property~\ref{it:B}, consider an edge $uv\in B$, for which by
definition of $B$ there is a $P_{st}\in \mc{P}$ and a  $u\to v$ subpath $P$ of $P_{st}$ fully contained in $N$. The path $P$
may use the edges of both $N_1$ and~$N_2$. By definition of~$B$, at
least one of $u$ and $v$ is in $U=U_1$, while $U_1\cap U_2$ separates
$N_1$ and~$N_2$. This means that we can partition $P$ into subpaths
such that for each subpath $P'$ there is an $h\in\{1,2\}$ for which
$P'$ uses only the edges of $N_h$, has endpoints in $(U_h\times
U_h)\cup(U_h\times R_{in})\cup(R_{out}\times U_h)$, and internal
vertices outside of $U_h$. If both endpoints $u',v'$ of $P'$ are
in~$U_h$, then the $U_h$-projection of $N$ will contain a
corresponding edge $u'v'$, which is also contained in
$B_h$. Similarly, $B_h$ will contain $u'v'$ if one of $u'$ and $v'$
lies in $R_{in}$ or $R_{out}$. Thus in any of these cases,  property~\ref{it:B} for~$N'_h$ implies that $N'_h$ contains a
$u'\to v'$ path. By replacing
each subpath $P'$ of $P$ with a path of $N'_1$ or $N'_2$ having the
same endpoints, we obtain that $N'$ also contains a $u\to v$ path, and
thus $N'$ has property~\ref{it:B}.

Finally, let us verify that $N'$ satisfies
property~\ref{it:R}. Consider an out-star $S_j$ with a leaf $\ell\in
L_j\cap Q$ and let $h\in\{1,2\}$ such that $\ell\in Q_h$. We have to
show that $N'$ contains an $A_j\cup \{r_j\}\to \ell$ path.  As $N'_h$
satisfies the entry $T[i_h,Q_h,U_h,\mc{A}_h,B_h,I_h]$, we know that
$N'_h$ has an $A^h_j\cup \{r_j\}\to \ell$ path $P'_h$. If $P'_h$ starts
in $r_j$, then we are done: then the path $P'_h$ shows that the
supergraph $N'$ of $N'_h$ contains a path from $A_j\cup\{r_j\}$ to
$\ell$. Suppose therefore that $P'_h$ starts in a vertex $v\in
A^h_j$. When defining the type of $N_h$, vertex $v$ was added to the
set $A^h_j$ because there is a leaf $\ell^*\in L_j\cap Q_h$ such that the
maximal suffix of the path $P_{r_j\ell^*}\in \mc{P}$ starts in
$v$. Suppose that the maximal suffix of $P_{r_j\ell^*}$ in $N$ starts
in some vertex $w\in A_j\cup \{r_j\}$, and let $Q$ be the $w\to v$
subpath of $P_{r_j\ell^*}$. We claim that, with an argument similar to
the previous paragraph, the $w\to v$ subpath of $Q$ can be turned into
a path $Q'$ of $N'$. Indeed, $Q$ can be partitioned into subpaths such
that for each subpath there is an $h^*\in\{1,2\}$ for which the
subpath uses only the edges of $N_{h^*}$, has endpoints in $(U_{h^*}\times
U_{h^*})\cup (R_{out}\times U_{h^*})$, and internal vertices outside of $U_h$.
Property~\ref{it:B} for~$N'_{h^*}$ implies that each such subpath can
be replaced by a path of $N'_{h^*}$, which proves the existence of the
required $w\to v$ subpath $Q'$ of $N'$. Then the concatenation of $Q'$
and $P'_h$ gives an $A_j\cup \{r_j\}\to \ell$ walk in $N'$, what we
had to show.  The case when $S_j$ is an in-star is symmetric.

 In conclusion, $N'$ has properties~\ref{it:vert} 
to~\ref{it:R} and satisfies $T[i,Q,U,I,B,\mc{A}]$.
\cqed\end{proof}

To conclude the proof, we need to show that the algorithm stores a network with 
cost at most that of $N$ in the entry $T[i,Q,U,I,B,\mc{A}]$.  Let 
$\gamma(E(\widetilde N))$ denote the total cost of all edges of a 
network~$\widetilde N$. As \autoref{clm:union} shows, the algorithm considers at some point $N'=N'_1\cup N'_2$ as a potential candidate for the entry $T[i,Q,U,I,B,\mc{A}]$, hence in the end the algorithm stores in this entry a partial solution with cost not more than $\gamma(E(N'))$. Thus the only thing we need to show is that $\gamma(E(N'))\le \gamma(E(N))$. As $U_1\cap U_2$ separates $N_1$ and $N_2$, the only edges that $N_1$ and $N_2$ can share are the edges in $U_1\cap U_2$, that is, $\gamma(E(N))=\gamma(E(N_1))+\gamma(E(N_2))-\gamma(E(N[U_1\cap U_2]))$.
By property~\ref{it:I}, every edge of $N[U_1\cap U_2]$ appears in both $N'_1$ 
and $N'_2$. This means that $N'_1$ and $N'_2$ share at least this set of edges 
(they can potentially share more edges outside of~$U_1\cap U_2$). Therefore, we 
have $\gamma(E(N'))\leq \gamma(E(N'_1))+\gamma(E(N'_2))-\gamma(E(N[U_1\cap 
U_2]))=\gamma(E(N))$, what we had to show.
\end{proof}

Since an optimum solution is minimal, we may set $M$ to an optimum solution to 
$H$ in \autoref{lem:alg-tw}. If we also set $D=D_M$ and $Q=R$ in the lemma we 
get that $A_j=\emptyset$ for each $j\in\{1,\ldots,c\}$. Any entry of the table 
for which $Q=R$ and $A_j=\emptyset$ for each $A_j\in\mc{A}$ contains a feasible 
solution to pattern $H$ or is empty, due to property~\ref{it:R}. Hence if $M$ 
has treewidth $\omega$ and $H$ is the union of $c$ in- and out-stars, by 
\autoref{lem:alg-tw} there is an $i$ such that entry $T[i,Q,U,I,B,\mc{A}]$ 
will contain a feasible network with cost at most that of $M$, i.e., an optimum 
solution to $H$. By searching all entries of the table for which $Q=R$ and 
$A_j=\emptyset$ for each $A_j\in\mc{A}$ we can thus find the optimum solution to 
$H$. 

\paragraph{Bounding the runtime.} 
The number of entries of the table $T$ is bounded by the number of possible 
values for $i$, sets $Q$, $U$, $A_1$, $\ldots$, $A_c$, $B$, and graphs $I$. For 
$i$ there are at most $n$ possible values. As $Q\subseteq R$ and $|R|=k$, there 
are $2^k$ possible sets $Q$, and since $U\subseteq V(G)$ with $|U|\leq\omega+1$ 
and $|V(G)|=n$, there are $n^{O(\omega)}$ subsets $U$. For each $A_j\in\mc{A}$ 
we choose a subset of~$U$ and thus there are $2^{\omega+1}$ such sets. The total 
number of sets $\mc{A}$ is thus $2^{c(\omega+1)}$, given a set $U$. The set $B$ 
contains  at most $\omega+1$ edges for each of the $c$ star centers in 
$R_{in}\cup R_{out}$ (i.e.,~a total of~$2^{c(w+1)}$ possibilities) and at most 
$c\omega$ edges induced by~$U$. As there are at most $2{\omega+1 \choose 2}< 
(\omega+1)^2$ possible edges induced by $U$, the number of possibilities for $B$ 
to contain at most $c\omega$ such edges is at most 
$\sum_{l=0}^{c\omega}{(\omega+1)^2 \choose l} \leq (\omega+1)^{2c\omega}$. Thus 
the total number of possible sets~$B$, given a fixed~$U$, is 
$2^{O(c\omega\log\omega)}$. The graph $I$ has at most $c\omega$ edges incident 
to the vertices of~$U$, and thus as before there are at most 
$2^{O(c\omega\log\omega)}$ possible such graphs. Therefore the number of entries 
in the table $T$ is $2^{O(k+c\omega\log\omega)}n^{O(\omega)}$.

In case $i\leq\omega+1$, the algorithm just checks whether $I$ has 
properties~\ref{it:vert} to~\ref{it:R}, and each of these checks can be done in 
time polynomial in $\omega$. In case $i>\omega+1$, every pair of entries with 
$i_1,i_2<i$ needs to be considered in order to form the union of the stored 
partial solutions. For the union, properties~\ref{it:vert} to~\ref{it:R} can be 
checked in polynomial time. Thus the time to compute an entry is 
$2^{O(k+c\omega\log\omega)}n^{O(\omega)}$, from which the total running time 
follows as $c\leq 2\tau$. This completes the proof of \autoref{crl:alg}.

\section{Characterizing the hard cases}
\label{sec:hard}

We now turn to proving the second part of \autoref{thm:main}, i.e., that
$\mc{H}$-DSN is W[1]-hard for every class~$\mc{H}$ where the patterns are not 
transitively equivalent to almost-caterpillars.
As we will see later, we need the minor technical requirement that the class 
$\mc{H}$ is recursively enumerable, in order to prove the following hardness 
result via reductions.

\begin{thm}\label{thm:main-hard}
Let $\mc{H}$ be a recursively enumerable class of patterns for which there are 
no constants $\lambda$ and $\delta$  such that 
$\mc{H}\subseteq\mc{C}^*_{\lambda,\delta}$. Then the problem $\mc{H}$-DSN is 
\textup{W[1]}-hard for parameter~$k$.
\end{thm}

A major technical simplification is to assume that the class $\mc{H}$ is closed 
under identifying terminals and transitive equivalence. As we show in 
\autoref{sec:reductions}, this assumption is not really restrictive: it is 
sufficient to prove hardness for the closure of $\mc{H}$ under identification 
and transitive equivalence, since any W[1]-hardness result for the closure can 
be transferred to $\mc{H}$. 
For classes closed under these operations, it is possible to give an elegant 
characterization of the classes that are not almost-caterpillars. There are only 
a few very specific reasons why a class $\mc{H}$ is not in 
$\mc{C}^*_{\lambda,\delta}$ for any $\lambda$ and $\delta$: either $\mc{H}$ 
contains every directed cycle, or $\mc{H}$ contains every ``pure diamond,'' or 
$\mc{H}$ contains every ``flawed diamond'' (see \autoref{sec:obstr-sccs-diam} 
for the precise definitions). Then in \autoref{sec:reductions-1}, we provide 
a W[1]-hardness proof for each of these cases, completing the hardness part of 
\autoref{thm:main}.

\subsection{Closed classes}\label{sec:reductions}

\begin{figure}[t]
\centering
\includegraphics[scale=1]{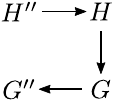}
\caption{A schematic representation of the reduction in \autoref{lem:closure}.}
\label{fig:red-closure}
\end{figure}

We define the  operation of \emph{identifying terminals} in the following way: 
given a partition $\mc{V}$ of the vertex set $V(H)$ of a pattern graph $H$, each 
set $W\in\mc{V}$ is identified with a single vertex of $W$, after which any 
resulting isolated vertices and self-loops are removed, while parallel edges 
having the same head and tail are replaced by only one of these copies. A 
class of patterns is \emph{closed} under this operation if for any pattern $H$ 
in the class, all patterns that can be obtained by identifying terminals are 
also in the class. Similarly, we say that a class $\mc{H}$ is closed under 
transitive equivalence if whenever $H$ and $H'$ are two transitively equivalent 
patterns such that $H\in\mc{H}$, then $H'$ is also in $\mc{H}$. The closure of 
the class $\mc{H}$ under identifying terminals and transitive equivalence is the 
smallest closed class $\mc{H'}\supseteq \mc{H}$. It is not difficult to see that 
any member of the closure can be obtained by a single application of identifying 
terminals and a subsequent replacement with a transitively equivalent pattern.

The following lemma shows that if we want to prove W[1]-hardness for a class, 
then it is sufficient to prove hardness for its closure. More precisely, due to 
a slight technicality, the actual statement we prove is that it is sufficient 
to prove W[1]-hardness for a decidable subclass of the closure.

\begin{lem}\label{lem:closure}
  Let $\mc{H}$ be a recursively enumerable class of patterns, let
  $\mc{H}'$ be the closure of~$\mc{H}$ under identifying
  terminals and transitive equivalence, and let $\mc{H}''$ be a decidable 
subclass of~$\mc{H}'$. There is a parameterized reduction from $\mc{H}''$-DSN to
  $\mc{H}$-DSN with parameter $k$.
\end{lem}
\begin{proof}%
Let us fix an enumeration of the graphs in $\mc{H}$, and consider the function 
$g:\mc{H'}\to\mathbb{N}$ that maps any graph $H'\in\mc{H'}$ to the number of 
vertices of the first graph $H\in\mc{H}$ in the enumeration such that $H'$ can 
be obtained from $H$ by identifying terminals and transitive equivalence. We 
define $f(k)=\max\{g(H'')\mid \textup{$H''\in\mc{H''}$ and $|V(H'')|=k$}\}$ to 
be the largest size of such an $H\in\mc{H}$ for any graph of 
$\mc{H''}\subseteq\mc{H}'$ with $k$ vertices. Note that $f$ only depends on the 
parameter $k$ and the classes $\mc{H}$ and $\mc{H''}$. Furthermore, $f$ is a 
computable function: as $\mc{H''}$ is decidable, there is an algorithm that 
first computes every $H''\in\mc{H''}$ with $k$ vertices, and then starts 
enumerating $\mc{H}$ to determine $g(H'')$ for each such $H''$.

For the reduction (see \autoref{fig:red-closure}), let an instance of 
$\mc{H''}$-DSN be given by an edge-weighted directed graph~$G''$ and a pattern 
$H''\in\mc{H''}$. We first enumerate patterns $H\in\mc{H}$ until finding one 
from which $H''$ can be obtained by identifying terminals and transitive 
equivalence. The size of $H$ is at most $f(k)$ if $k=|V(H'')|$, and checking 
whether a given pattern of $\mc{H''}$ can be obtained from $H$ by identifying 
terminals can be done by brute force. Thus the time needed to compute $H$ 
depends only on the parameter $k$.

Let $W_t\subseteq V(H)$ denote the set of vertices that are identified with 
$t\in V(H')$ to obtain~$H''$. In $G''$ we add a strongly connected graph on 
$W_t$ with edge weights~$0$ for every $t\in V(H'')$, by first adding the 
vertices $W_t\setminus\{t\}$ to $G''$ and then forming a cycle of the vertices 
of $W_t$. It is easy to see that we obtain a graph $G$ for which any solution 
$N\subseteq G$ to~$H$ corresponds to a solution $N''\subseteq G''$ to $H''$ of 
the same cost, and vice versa. Since the new parameter $|V(H)|$ is at most 
$f(k)$ and the size of $G$ is larger than the size of $G''$ by a factor bounded 
in terms of~$f(k)$, this is a proper parametrized reduction from  $\mc{H''}$-DSN 
to  $\mc{H}$-DSN.
\end{proof}

\subsection{Obstructions: SCCs and diamonds}
\label{sec:obstr-sccs-diam}

\begin{figure}[t]
{a)\includegraphics[scale=1]{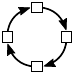}}\hfill
{b)\includegraphics[scale=1]{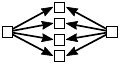}}\hfill
{c)\includegraphics[scale=1]{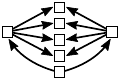}}\hfill
{d)\includegraphics[scale=1]{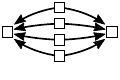}}\hfill
{e)\includegraphics[scale=1]{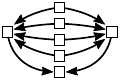}}
\caption{The obstruction appearing in \autoref{thm:char}: a) a directed cycle 
of length $4$, b) a pure $4$-out-diamond, c)~a~flawed $4$-out-diamond, d) a 
pure $4$-in-diamond, e) a flawed $4$-in-diamond.}
\label{fig:diamond}
\end{figure}

 To show the hardness 
 for a closed class that is not the subset of 
 $\mc{C}^*_{\lambda,\delta}$ for any $\lambda$ and $\delta$, we will 
characterize such a class in terms 
 of the occurrence of arbitrarily large cycles, and another class of patterns 
called ``diamonds'' 
 (cf.~\autoref{fig:diamond}). 

\begin{dfn}\label{dfn:diamond}
A \emph{pure $\alpha$-diamond} graph is constructed as follows. Take a vertex 
set $L$ of size $\alpha\geq 1$, and two additional vertices $r_1$ and $r_2$. 
Now add edges such that $L$ is the leaf set of either two in-stars or two 
out-stars $S_1$ and $S_2$ with roots $r_1$ and $r_2$, respectively. If we add 
an additional vertex $x$ with edges $r_1x$ and $r_2x$ if $S_1$ and $S_2$ are 
in-stars, and edges $xr_1$ and $xr_2$ otherwise, the resulting graph is a 
\emph{flawed $\alpha$-diamond}. We refer to both pure $\alpha$-diamonds and 
flawed $\alpha$-diamonds as \emph{$\alpha$-diamonds}. If $S_1$ and $S_2$ are 
in-stars we also refer to the resulting $\alpha$-diamonds as 
\emph{in-diamonds}, and otherwise as \emph{out-diamonds}.
\end{dfn}

The goal of this section is to prove the following useful characterization 
precisely describing classes that are not almost-caterpillars.

\begin{lem}\label{thm:char}
Let $\mc{H}$ be a class of pattern graphs that is closed under identifying 
terminals and transitive closure. Exactly one of the following statements is 
true:
\begin{itemize}
\item $\mc{H}\subseteq \mc{C}^*_{\lambda,\delta}$ for some constants $\lambda$ 
and $\delta$.
\item $\mc{H}$ contains every directed cycle, or every pure in-diamond, or
every pure out-diamond, or every flawed in-diamond, or every flawed out-diamond.
\end{itemize}
\end{lem}

For the proof of \autoref{thm:main-hard}, we only need the fact that at least 
one of these two statements hold: if the class $\mc{H}$ is not in 
$\mc{C}^*_{\lambda,\delta}$, then we can prove hardness by observing that 
$\mc{H}$ contains one of the hard classes. For the sake of completeness, we 
give a simple proof that the two statements cannot hold simultaneously (note 
that it is sufficient to require closure under transitive equivalence for this 
statement to hold). 

\begin{lem}\label{lem:if}
Let $\mc{H}$ be a class of pattern graphs that is closed under transitive equivalence. If there are constants $\lambda$ and 
$\delta$ such that $\mc{H}\subseteq\mc{C}^*_{\lambda,\delta}$, then $\mc{H}$ 
cannot contain a pure or flawed $\alpha$-diamond or a cycle of length $\alpha$ 
for any $\alpha>2\delta+\lambda$.
\end{lem}
\begin{proof}%
  Suppose first that there is a pattern
  $H\in\mc{C}^*_{\lambda,\delta}$ that is a cycle of length
  $\alpha$. There is a pattern $H'\in\mc{C}_{\lambda,\delta}$ that is
  transitively equivalent to $H$. Clearly, any graph that is
  transitively equivalent to a directed cycle is strongly connected,
  which then also applies to $H'$. Recall that according to
  \autoref{dfn:caterpillar} there is a set of edges $F\subseteq E(H')$
  of size at most $\delta$ for which the remaining edges
  $E(H')\setminus F$ span a $\lambda_0$-caterpillar $C$ for some
  $\lambda_0\leq\lambda$. That is, $C$ consists of $\lambda_0$
  vertex-disjoint stars for which their roots are joined by a
  path. Since every vertex of a strongly connected graph must have in-
  and out-degree at least $1$, any leaf of a star of $C$ can only be
  part of an SCC if it is incident to some edge of~$F$. Hence if $H$
  was strongly connected, then for every leaf of $C$ there would be an
  additional edge in $F$. This however would mean that $H$ contained
  at most $2\delta+\lambda$ vertices: for each edge of $F$ the two
  incident vertices, which include the leaves of the caterpillar, and
  $\lambda_0\leq\lambda$ roots of stars. Hence $\alpha\leq
  2\delta+\lambda$.

  Suppose now that there is a pattern $H\in\mc{C}^*_{\lambda,\delta}$
  that is an $\alpha$-diamond, and a pattern
  $H'\in\mc{C}_{\lambda,\delta}$, which is transitively equivalent to
  $H$. Let $r_1$ and $r_2$ be the two roots of the diamond $H$, and
  let us denote by $r_1$ and $r_2$ the corresponding two vertices in
  $H'$ as well.  It is easy to see from \autoref{dfn:diamond} that $H'$
  contains an $\alpha$-diamond as a subgraph, possibly in addition to
  some edges that connect the vertex $x$ with some of the leaves in
  $L$, in case of a flawed $\alpha$-diamond. This means that $r_1$ and
  $r_2$ have degree at least $\alpha$ in $H'$ as well.  Let $F$ be a set of
  at most $\delta$ edges such that $E(H')\setminus F$ span a $\lambda_0$-caterpillar
  $C$ for some $\lambda_0\leq\lambda$. It is not possible that both $r_1$ and $r_2$ are on the spine of the caterpillar: then there would be a directed path from one to the other, which is not the case in the diamond $H$. Assume without loss of generality that $r_1$ is not on the spine of the caterpillar. Then $r_1$ has degree at most 1 in $E(H')\setminus F$ and hence degree at most $|F|+1\le \delta+1$ in $H'$. As we observed, $r_1$ has degree at least $\alpha$ in $H'$, it follows that $\alpha\le \delta+1$. 
\end{proof}

Showing that at least one of the two statements of \autoref{thm:char} hold is 
not as easy to prove.  First, the following two lemmas show how a large cycle or 
a large diamond can be identified if certain structures appear in a pattern. The 
main part of the proof is to show that if $\mc{H}$ contains patterns that are 
arbitrarily far from being a caterpillar, then one of these two lemmas can be 
invoked (see \autoref{lem:only-if}). For the next lemma we define a 
\emph{matching} of a graph as a subset $M$ of its edges such that no two edges 
of $M$ share a vertex. 

\begin{lem}\label{lem:match}
Let $\mc{H}$ be a class of pattern graphs that is closed under identifying 
terminals and transitive closure. If some $H\in\mc{H}$ contains a matching of size $\alpha$, then $\mc{H}$ contains a directed cycle of length $\alpha$.
\end{lem}
\begin{proof}
A matching $e_1$, $\dots$, $e_\alpha$ of $\alpha$ edges can be transformed  
into a cycle of length $\alpha$ by identifying the head of $e_i$ and tail of 
$e_{i+1}$ (and the head of $e_\alpha$ with the tail of~$e_1$). All remaining 
vertices of $H$ that do not belong to the cycle can then be identified with any 
vertex of the cycle, so that the resulting graph consists of the cycle and some 
additional edges. Since $\mc{H}$ is closed under identifying terminals, this 
graph is contained in $\mc{H}$ if $H$ is. As this graph is strongly connected 
and $\mc{H}$ is closed also under transitive equivalence, we can conclude that 
$\mc{H}$ contains a cycle of length $\alpha$.
\end{proof}

Next we give a sufficient condition for the existence of large diamonds. We say 
that an edge $uv$ of a graph $H$ is {\em transitively non-redundant} if there is 
no $u\to v$ path in~$H\setminus uv$.

\begin{lem}\label{lem:find-diamond}
Let $\mc{H}$ be a class of pattern graphs that is closed under identifying 
terminals and transitive equivalence. Let $H\in\mc{H}$ be a pattern graph that 
contains two out-stars (or two in-stars) $S_1$ and $S_2$ as induced subgraphs, 
with at least $\alpha$ edges each and roots $r_1$ and $r_2$, respectively, 
such that $r_1\neq r_2$. If
\begin{enumerate}
\item $H$ contains neither a path from $r_1$ to 
$r_2$, nor from $r_2$ to $r_1$,
\item the leaves of $S_1$ and $S_2$ have out-degree 0 (if $S_1$ and $S_2$ are 
out-stars) or in-degree 0 (if $S_1$ and $S_2$ are in-stars), and 
\item the edges of the stars are transitively non-redundant,
\end{enumerate}
then $\mc{H}$ contains an $\alpha$-diamond.
\end{lem}
\begin{proof}%
We only consider the case when $S_1$ and $S_2$ are out-stars, as the other case 
is symmetric. Let $T_1\subseteq S_1$ and $T_2\subseteq S_2$ be two out-stars 
with exactly $\alpha$ edges and roots $r_1$ and~$r_2$, respectively. We 
construct an $\alpha$-diamond starting from $T_1$ and $T_2$, and using the 
following partition of $V(H)$. Let $\{s_1,\ldots,s_\alpha\}$ and 
$\{t_1,\ldots,t_\alpha\}$ denote the leaf sets of $T_1$ and~$T_2$.
These sets 
may intersect, but we may order them in a way that $i=j$ holds whenever 
$s_i=t_j$. Define $Y_1\subseteq V(H)\setminus V(T_1\cup T_2)$ and $Y_2\subseteq 
V(H)\setminus V(T_1\cup T_2)$ to be the reachability sets of $r_1$ and $r_2$, 
i.e., they consist of those vertices $w$ that do not belong to $T_1$ or~$T_2$, 
and for which there is a path in $H$ to $w$ from $r_1$ or~$r_2$, respectively. 
We partition all vertices of $H$ outside of the two stars $T_1$ and $T_2$ into 
the set $W_1=Y_1\setminus Y_2$ reachable from only $r_1$, the set 
$W_2=Y_2\setminus Y_1$ reachable from only $r_2$, the set $W=Y_1\cap Y_2$ 
reachable from both $r_1$ and~$r_2$, and the set $U=V(H)\setminus(Y_1\cup Y_2)$ 
reachable from neither $r_1$ nor $r_2$. 

To obtain an $\alpha$-diamond,  we identify  for each $i\in\{1,\ldots,\alpha\}$ 
the leaves $s_i$ and $t_i$, and call the resulting vertex~$\ell_i$. We also 
identify every vertex of $W_1$ with $r_1$, every vertex of $W_2$ with~$r_2$, and 
all vertices in $W$ with the vertex~$\ell_1$. If there is a vertex $x$ in $U$ 
for which in $H$ there is a path to some vertex in $W_1\cup\{r_1\}$, and there 
is a vertex $x'$ in $U$ (which may be equal to $x$) with a path to a vertex in 
$W_2\cup\{r_2\}$, then we identify each vertex in $U$ with~$x$. If there is no 
path from any vertex of $U$ to a vertex of $W_2\cup\{r_2\}$, but for some vertex 
in $U$ there is a path to~$W_1\cup\{r_1\}$, we identify every vertex of $U$ 
with~$r_1$. Otherwise, all vertices of $U$ are identified with $r_2$. We claim 
that the resulting graph $D$ is a pure $\alpha$-diamond if the pair~$x,x'$ does 
not exist, and transitively equivalent to a flawed $\alpha$-diamond otherwise.

The graph $D$ clearly contains a pure $\alpha$-diamond as a subgraph, due to the 
stars $T_1$ and~$T_2$. If the pair $x,x'\in U$ exists it also contains a flawed 
$\alpha$-diamond, since the two paths from $x$ to $W_1\cup\{r_1\}$ and from $x'$ 
to $W_2\cup\{r_2\}$ result in edges $xr_1$ and $xr_2$ after identifying $W_1$ 
with~$r_1$, $W_2$ with $r_2$, and $U$ with~$x$. There may be edges $x\ell_i$ in 
$D$ for some $i\in\{1,\ldots,\alpha\}$, but these are transitively implied by 
the path consisting of the edges $xr_1$ and $r_1\ell_i$. Hence if no other edges 
exist in $D$, it is transitively equivalent to a (pure or flawed) 
$\alpha$-diamond.

By assumption the out-degree of each leaf of the out-stars $T_1$ and $T_2$ is 
$0$. Hence for $i\geq 2$, none of the above identifications can add an edge with 
a vertex~$\ell_i$ as its tail. For~$\ell_1$ it could possibly happen that an 
edge with $\ell_1$ as its tail was introduced when identifying $W$ with this 
vertex. The head of such an edge in~$D$ would be either some $\ell_i$ with 
$i\geq 2$, $r_1$, $r_2$, or $x$ if it exists. This would mean that in $H$ there 
is an edge $yz$ with $y\in W$ and $z\in\{s_i,t_i,r_1,r_2\}\cup U$. By definition 
of $W$, in $H$ there is both a path from $r_1$ and from $r_2$ to $y$, and 
furthermore none of these paths contains $s_i$ or $t_i$, as these vertices have 
out-degree~$0$. Assume first that $z=s_i$, in which case the $r_1\to y$ path 
together with the edge $ys_i$ form a path not containing the edge $r_1s_i$. 
However this contradicts the assumption that $r_1s_i$ is transitively 
non-redundant. Similarly, it cannot be that $z=t_i$, since otherwise $r_2t_i$ 
would be transitively redundant. If $z=r_1$, then there is a path from $r_2$ 
to~$r_1$ through $y$, which is excluded by our assumption that no such path 
exists. Symmetrically it can also not be that $z=r_2$. The only remaining option 
is that~$z\in U$. However this is also excluded by definition of $U$, as 
otherwise there would be a path from $r_1$ to $U$ through~$y$. Consequently, the 
out-degree of~$\ell_i$ in $D$ is~$0$ for every~$i\in\{1,\ldots,\alpha\}$.

In case the pair $x,x'$ exists in $H$, it is not hard to see that there is no 
edge in~$D$ with $x$ as its head: by definition of $U$ there is no edge $yz$ in 
$H$ with $y\notin U$ and $z\in U$, as in $H$ there are no paths from $r_1$ or 
$r_2$ to any vertex of $U$, while every vertex outside of $U$ is reachable from 
$r_1$ or~$r_2$. Thus it remains to argue that there is no edge between $r_1$ and 
$r_2$ in $D$. If the pair~$x,x'$ does not exist, $U$ is identified with either 
$r_1$ or $r_2$. The former only happens if there is no vertex in $U$ with a path 
to $r_2$, while the latter only happens if no such vertex with a path to $r_1$ 
exists. Hence identifying $U$ with either $r_1$ or $r_2$ does not add an edge 
between $r_1$ and $r_2$. Note that in $H$ there cannot be an edge $yz$ with 
$y\in W_1$ and $z\in W_2$, since otherwise $z\in Y_1$, which contradicts the 
definition of $W_2$. Analogously, no edge $yz$ with $y\in W_2$ and $z\in W_1$ 
exists either. Consequently, identifying $W_2$ with $r_2$ and $W_1$ with $r_1$ 
does not add any edge between $r_1$ and $r_2$ to $D$. This concludes the proof 
since no additional edges exist in $D$.
\end{proof}

To show that at least one of the two statements of \autoref{thm:char}
hold, we prove that if the second statement is false, then the first
statement is true. Observe that if a class closed under
identifications contain an $\alpha$-cycle or $\alpha$-diamond, then it
contains every cycle or diamond of smaller size.  Thus what we need to show is that if $\mc{H}$ does not contain all cycles (i.e.,
there is an $\alpha_1$ such that $\mc{H}$ contains no cycle larger
than $\alpha_1$), $\mc{H}$ does not contain all pure out-diamonds
(i.e., there is an $\alpha_2$ such that $\mc{H}$ contains no pure
out-diamond larger than $\alpha_2$), etc., then $\mc{H}\subseteq
\mc{C}^*_{\lambda,\delta}$ for some constants $\lambda$ and
$\delta$. In other words, if we let $\alpha$ to be the maximum of
$\alpha_1$, $\alpha_2$, etc., then we may assume that $\mc{H}$
contains no pure or flawed $\alpha$-diamond or cycle of length
$\alpha$, and we need to prove $\mc{H}\subseteq
\mc{C}^*_{\lambda,\delta}$ under this assumption. Thus the following
lemma completes the proof of \autoref{thm:char}.
\begin{lem}\label{lem:only-if}
Let $\mc{H}$ be a class of pattern graphs that is closed under identifying 
terminals and transitive equivalence. If for some integer $\alpha$ the class 
$\mc{H}$ contains neither a pure \mbox{$\alpha$-diamond}, flawed 
$\alpha$-diamond, nor a cycle of length $\alpha$, then there exist constants 
$\lambda$ and $\delta$ (depending on~$\alpha$) such that 
$\mc{H}\subseteq\mc{C}^*_{\lambda,\delta}$. 
\end{lem}
\begin{proof}
  Suppose that there is such an integer $\alpha$. Let $\lambda:=2\alpha$ and 
$\delta:=4\alpha^3+6\alpha^2$. Given any $H'\in\mc{H}$, we
  show how a transitively equivalent pattern
  $H\in \mc{C}_{\lambda,\delta}$ can be constructed, implying that
  $H'$ belongs to $\mc{C}^*_{\lambda,\delta}$.
A \emph{vertex cover} of a graph is a subset $X$ of its vertices such
that every edge is incident to a vertex of $X$. By
\autoref{lem:match}, $H'$ cannot contain a matching of size
$\alpha$. It is well-known that if a graph has no matching of size
$\alpha$, then it has a vertex cover of size at most $2\alpha$ (take
the endpoints of any maximal matching). Let us fix a vertex cover $X$ of $H'$ 
having size at most $2\alpha$.

To obtain $H$ from $H'$, we start with a graph $H$ on $V(H')$ having no edges 
and perform the following three steps. 
\begin{enumerate}
\item Let us take the transitive closure on the vertex set $X$ in~$H'$, i.e., 
let us introduce into $H$ every edge $uv$ with $u,v\in X$ such that there is a 
$u\to v$ path in $H'$.
\item Let us add 
all edges $uv$ of $H'$ to 
$H$ for which $u\notin X$ or $v\notin X$.
\item Fixing an ordering of the edges introduced in step 2, we remove 
transitively redundant edges: 
following this order, we subsequently remove those edges $uv$ for which there is a 
path from $u$ to $v$ in the remaining graph $H$ that is not the edge $uv$ 
itself (we emphasize that the edges with both endpoint in $X$ are not touched in this step). 
\end{enumerate}
It is clear that $H$ is transitively equivalent to $H'$, hence $H\in \mc{H}$.  Note that
$X$ is a vertex cover of $H$ as well, and hence its complement
$I=V(H)\setminus X$ is an \emph{independent set}, i.e., no two vertices of $I$ 
are adjacent. Let $E_I\subseteq E(H)$ be
the set of edges between $X$ and $I$.  In the rest of the proof, we
argue that the resulting pattern $H$ belongs to~$\mc{C}_{\lambda,\delta}$. We 
show that $H$ can be decomposed into a path
$P=(v_1,\dots, v_{\lambda_0})$ in $X$, a star $S_{v_i}$ centered at
each $v_i$ using the edges in $E_I$, and a small set of additional edges. This 
small set of additional edges is constructed in three steps, by considering a 
sequence of larger and larger sets $F_1\subseteq F_2 \subseteq F_3$.

As $E_I$ consists of edges between $X$ and $I$, it can be partitioned into a 
set of stars with roots in $X$. The following claim shows that almost all of 
these edges are directed towards $X$ or almost all of them are directed away 
from $X$.

\begin{claim}\label{lem:in-stars}
Either there are less than $2\alpha^2$ edges $uv$ in $E_I$ with head in $X$, or 
less than $2\alpha^2$ edges $uv$ in $E_I$ with tail in $X$.
\end{claim}
\begin{proof}%
  Assume $H$ contains an in-star $S_{in}$ and an out-star $S_{out}$ as
  subgraphs, each with $\alpha-1$ edges from $E_I$ and roots
  in~$X$. Let $\{s_1,\ldots,s_{\alpha-1}\}$ and
  $\{t_1,\ldots,t_{\alpha-1}\}$ denote the leaf sets of $S_{in}$ and
  $S_{out}$, respectively. These sets may intersect, but we may order
  them in a way that $i=j$ holds whenever $s_i=t_j$.  First
  identifying the roots of $S_{in}$ and~$S_{out}$, and then $s_i$ and
  $t_i$ for each $i\in\{1,\ldots,\alpha-1\}$, we obtain a strongly
  connected subgraph on $\alpha$ vertices. Further identifying any
  other vertex of $H$ with an arbitrary vertex of this subgraph yields
  a strongly connected graph on $\alpha$ vertices. This graph is
  transitively equivalent to a cycle of length $\alpha$, a
  contradiction to our assumption that $\mc{H}$ does not contain any
  such graph. Consequently, either all in-stars spanned by subsets of
  $E_I$ with roots in~$X$ have size less than $\alpha-1$, or all such
  out-stars have size less than~$\alpha-1$. Assume the former is the
  case, which means that every edge $uv\in E_I$ with $v\in X$ is part
  of an in-star of size less than $\alpha-1$.  Since $X$ contains less
  than $2\alpha$ vertices, there are less than $2\alpha^2$ such
  edges. The other case is analogous.  \cqed
\end{proof}

Assume that the former case of \autoref{lem:in-stars} is true, so that
the number of edges in $E_I$ with heads in $X$ is bounded by
$2\alpha^2$; the other case can be handled symmetrically. We will use
the out-stars spanned by $E_I$ for the caterpillar, which means that
we obtain an out-caterpillar.  We use the set $F_1$ to account for the edges in 
$E_I$ with heads in $X$. Additionally, we will also introduce into $F_1$ those 
edges in $E_I$ with tails in $X$ that are adjacent to an
edge of the former type. Formally, for any edge $uv\in E_I$ with
$v\in X$, we introduce into $F_1$ every edge of $E_I$ incident to $u$.
After this step, $F_1$ contains less than $4\alpha^3$ edges, since there
are less than $2\alpha^2$ edges $uv\in E_I$ with $v\in X$ and $u$ can
only be adjacent to vertices in $X$, which has size less than
$2\alpha$.

For any vertex $v\in X$, let $S_v$ denote the out-star formed by the
edges of $E_I\setminus F_1$ incident to $v$. Let $X'\subseteq X$
contain those vertices $v\in X$ for which $S_v$ has at least $\alpha$
leaves.

\begin{claim}\label{cl:starpath}
For any two distinct $u,v\in X'$, at least one of $uv$ and $vu$ is in $H$, and 
the stars $S_u$ and $S_v$ are vertex disjoint.
\end{claim}
\begin{proof}%
  Suppose that there is no edge between $u$ and $v$. In step 1 of the
  construction of~$H$, we introduced any edge between vertices of $X$
  that appears in the transitive closure, so it also follows that
  there is no directed $u\to v$ or $v\to u$ path in $H'$ and hence in $H$. By
definition, the star $S_v$ and $F_1$ are edge disjoint, which implies that the out-degree of any leaf of the
  out-star $S_v$ is $0$ in $H$.  By step 3 of the construction of $H$,
  the edges of $S_u$ and $S_v$ are transitively non-redundant. Thus we
  can invoke \autoref{lem:find-diamond} to conclude that $\mc{H}$
  contains an $\alpha$-diamond, a contradiction.

Assume therefore that, say, edge $uv$ is in $H$. To prove that $S_u$ and 
$S_v$ are disjoint, suppose for a contradiction that they share a leaf $\ell$. 
But then the edges $uv$ and $v\ell$ show that the edge $u\ell$ is transitively 
redundant.  However, in step 3 of the construction of $H$,  we removed all 
transitively redundant edges incident to vertices not in $X$ to obtain~$H$, 
and~$\ell\notin X$, a contradiction.\cqed
\end{proof}

We extend $F_1$ to $F_2$ by adding all edges of stars $S_v$ with $v\in 
X\setminus X'$ to $F_2$. Since $X$ contains less than $2\alpha$ vertices and we 
extend $F_1$ only by stars with less than $\alpha$ edges, this step adds less 
than $2\alpha^2$ edges, i.e., $|F_2|\le |F_1|+2\alpha^2=4\alpha^3+2\alpha^2$.

By \autoref{cl:starpath}, $X'$ induces a \emph{semi-complete} directed
graph in $H$, i.e., at least one of the edges $uv$ and $vu$ exists for every
pair $u,v\in X'$. It is well-known that every semi-complete directed graph 
contains a Hamiltonian path (e.g., \cite[Chapter~10, Exercise~1]{MR2159259}), 
and so there is a path
$P=(v_1,\ldots,v_{\lambda_0})$ with
$\lambda_0=|X'|\leq 2\alpha=\lambda$ in $H$ on the vertices of
$X'$. We extend $F_2$ to $F_3$ by including any edge induced by
vertices of $X$ that is not part of $P$. 
There are less than $4\alpha^2$ such edges, and hence we have
$|F_3|\le |F_2|+4\alpha^2\le 4\alpha^3+6\alpha^2=\delta$.
The edges of $H$ not in $F_3$ span the path $P$ and disjoint out-stars
$S_{v_i}$ with $i\in\{1,\ldots,\lambda_0\}$, i.e., they form a
$\lambda_0$-caterpillar. This proves
that $H\in \mc{C}_{\lambda,\delta}$ and hence
$H'\in \mc{C}^*_{\lambda,\delta}$, what we had to show.
\end{proof}

\subsection{Reductions}
\label{sec:reductions-1}
\autoref{thm:char} implies that in order to prove \autoref{thm:main-hard}, we 
need W[1]-hardness proofs for the class of all directed cycles, the class of 
all pure in-diamonds, the class of all pure out-diamonds, etc. We provide 
these hardness proofs and then formally show that they imply 
\autoref{thm:main-hard}.

Let us first consider the case when $\mc{H}$ is the class of all
directed cycles. Recall that, given an arc-weighted directed graph $G$
and a set $R\subseteq V(G)$ of terminals, the \pname{Strongly
Connected Steiner Subgraph} (SCSS) problem asks for a minimum-cost subgraph 
that is strongly connected and contains every terminal in $R$. This problem is
known to be W[1]-hard parameterized by the number $k:=|R|$ of
terminals \cite{DBLP:journals/siamdm/GuoNS11}. We can reduce SCSS to an 
instance of DSN where the pattern $H$ is a directed cycle on $R$, which 
expresses the requirement that all the terminals are in the same strongly 
connected component of the solution. Thus the W[1]-hardness of SCSS immediately 
implies the W[1]-hardness of $\mc{H}$-DSN if $\mc{H}$ contains all directed 
cycles.

\begin{lem}[follows from \cite{DBLP:journals/siamdm/GuoNS11}]\label{lem:red-SCC}
If $\mc{H}$ is the class of directed cycles, then $\mc{H}$-DSN is 
\textup{W[1]}-hard parameterized by the number of terminals.
\end{lem}

Next we turn our attention to classes containing all diamonds. 
The following reductions are from the W[1]-hard \pname{Multicoloured Clique} 
problem~\cite{DBLP:journals/tcs/FellowsHRV09}, in which an undirected 
graph together with a partition $\{V_1,\ldots,V_k\}$ of its vertices into $k$ 
sets is given, such that for any $i\in\{1,\ldots,k\}$ no two vertices of $V_i$ 
are adjacent. The aim is to find a clique of size~$k$, i.e., a set of pairwise 
adjacent vertices $\{w_1,\ldots,w_k\}$ with $w_i\in V_i$ for each 
$i\in\{1,\ldots,k\}$.

\begin{lem}\label{lem:red-pdiam}
If $\mc{H}$ is the class of all pure out-diamonds, then $\mc{H}$-DSN is \textup{W[1]}-hard parameterized by the number of terminals. The same holds if $\mc{H}$ is the class of all pure in-diamonds.
\end{lem}
\begin{proof}%
We prove the statement only for out-diamonds, the other case is symmetric by 
reversing all directions of the edges in the description below. 

\textbf{Construction.} Consider an instance of \pname{Multicoloured Clique} with 
partition $\{V_1,\ldots,V_k\}$. For all indices $1\leq i<j\leq k$, we let 
$E_{ij}$ be the set of all edges connecting $V_i$ and $V_j$. We 
construct an instance of DSN where the pattern $H$ is a pure 
\mbox{$k(k-1)$-diamond}. Let $r_1$ and $r_2$ be the roots of the diamond and let 
$L=\{\ell_{ij}\mid 1\leq i,j\leq k\,\land\,i\neq j\}$ be the leaf set (so we 
have $|L|=k(k-1)$). The constructed input graph $G$ is the following (see 
\autoref{fig:red-diamond}). 
\begin{itemize}
 \item The terminals of $G$ are the terminals of $H$, i.e., $r_1$, $r_2$, and 
the vertices in $L$.
 \item For every $i\in\{1,\ldots,k\}$, we introduce into $G$ a
vertex $y_i$ representing $V_i$, and $k$ copies of each vertex $w\in V_i$, 
which we denote by $w_j$ for $j\in\{0,1,\ldots,k\}$ and~$j\neq i$. Also for 
all $1\leq i<j\leq k$, we introduce a vertex $z_{ij}$ representing $E_{ij}$, 
and a vertex $z_e$ for every edge~$e\in E_{ij}$.
 \item For every $i\in\{1,\ldots,k\}$, we add the edge $r_1y_i$, and for all 
$1\leq i<j\leq k$ the edge $r_2z_{ij}$.
 \item For every $i\in\{1,\ldots,k\}$ and $w\in V_i$, we add the edge $y_iw_0$, 
and for all $1\leq i<j\leq k$ and $e\in E_{ij}$, we add the edge $z_{ij}z_e$.
 \item For every $i,j\in\{1,\ldots,k\}$ with $i\neq j$ and $w\in V_i$, we add 
the edge~$w_0w_j$ and the edge~$w_j\ell_{ij}$.
 \item For all $1\leq i<j\leq k$ and $e\in E_{ij}$, for the vertex $w\in V_i$ 
incident to $e$, we add the edge~$z_ew_j$, and for the vertex $w\in V_j$ 
incident to $e$ we add the edge $z_ew_i$.
 \item Every edge of $G$ has cost $1$.
\end{itemize}

\begin{figure}[t!]%
\centering
\includegraphics[width=0.8\textwidth]{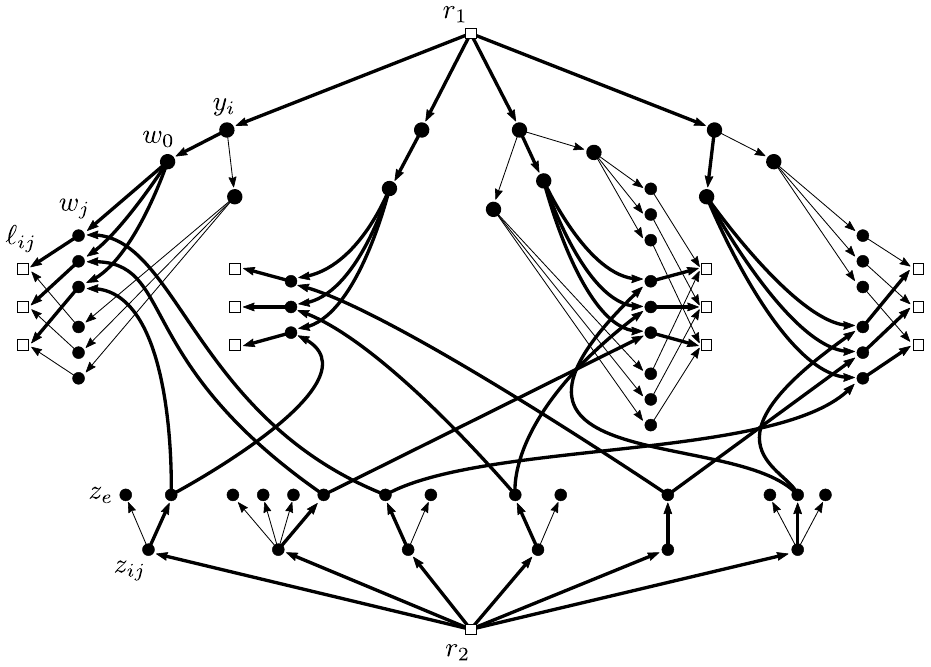}
\caption{The constructed graph in the reduction of \autoref{lem:red-pdiam} for 
an instance with $k=4$, $|V_1|=2$, $|V_2|=1$, $|V_3|=3$, and $|V_4|=2$. Squares 
are terminals, and circles are Steiner vertices. The thick edges indicate a 
solution encoding a clique. For better visibility, only the edges $z_ew_i$ of 
the solution are shown for any edge $e$ and incident vertex $w$.}
\label{fig:red-diamond}
\end{figure}

We prove that the instance to \pname{Multicoloured Clique} has a clique $K$ of 
size $k$, if and only if there is a solution $N$ to the pure $\alpha$-diamond 
$H$ in $G$ with cost at most $4k^2-2k$. Intuitively, such a solution $N$ will 
determine one vertex $w$ of $K$ for each $V_i$, since it can only afford to 
include the $k$ corresponding copies $w_j$ when connecting $r_1$ to $L$ through 
the vertex $y_i$ representing~$V_i$. At the same time $N$ will determine one 
edge $e$ of $K$ for each $E_{ij}$ by connecting $r_2$ to $L$ through one vertex 
$z_e$ for each vertex $z_{ij}$ representing~$E_{ij}$. These vertices $z_e$ are 
connected to the $k-1$ copies $w_j$ with $j>0$ of a vertex $w\in\bigcup_i V_i$ 
in such a way that $e$ must be incident to $w$ in order for the paths from 
$r_2$ in $N$ to reach $L$.

\textbf{Clique $\Rightarrow$ network.} We first show that a solution $N$ in $G$ 
of cost $4k^2-2k$ exists if the clique $K$ exists. For every 
$i\in\{1,\ldots,k\}$ the solution contains the edges $r_1y_i$ and $y_iw_0$, 
where $w$ is the vertex of $K$ in $V_i$. These edges add a cost of $2k$ to $N$. 
We also add all edges $w_0w_j$ for the $k-1$ additional copies $w_j$ with $j>0$ 
of each vertex $w$ of $K$, which adds a cost of $k(k-1)$. For each such copy 
$w_j$ we then connect to the terminal set $L$ by adding the respective edge 
$w_j\ell_{ij}$. Note that this will add an edge incident to each terminal of $L$ 
to $N$ and so $r_1$ is connected to every terminal of~$L$. At the same time the 
last step adds a cost of~$1$ for every terminal of $L$ to $N$, which sums up to 
$k(k-1)$. For all $1\leq i<j\leq k$ we connect $r_2$ to $z_{ij}$ in the solution 
$N$ via the edge $r_2z_{ij}$ at a cost of~$k \choose 2$. The clique $K$ contains 
one edge $e$ from every set $E_{ij}$, and we add the corresponding edges 
$z_{ij}z_e$ to $N$ at an additional cost of~$k \choose 2$. For any such edge $e$ 
the graph $G$ contains an edge $z_ew_j$ for the incident vertex $w\in V_i$ and 
an edge $z_ew_i$ for the other incident vertex $w\in V_j$. We also add these 
respective edges to the solution at a cost of~$2{k\choose 2}$. Since such an 
incident vertex $w\in V_i$ is part of the clique $K$, the respective copy $w_j$ 
is connected to the terminal $\ell_{ij}\in L$ in $N$. Moreover, every copy $w_j$ 
that is part of $N$ can be reached from the vertex $z_e$ in $N$ for the 
corresponding incident edge $e$ to $w$ in $K$. Hence $r_2$ is connected to every 
terminal of $L$ in~$N$, which means that $N$ is a solution to $H$ in $G$ with a 
total cost of $2k+2k(k-1)+4{k\choose 2}=4k^2-2k$.

\textbf{Network $\Rightarrow$ clique.} It remains to prove that any solution $N$ 
to $H$ in $G$ of cost at most $4k^2-2k$ corresponds to a clique $K$ of size $k$ 
in the input instance. If a solution to the pure $\alpha$-diamond $H$ exists 
in~$G$, then all terminals of $L$ are reachable from $r_1$ and from $r_2$ in 
$G$. We define the \emph{reachability set} $Y_v$ of a vertex~$v\in V(G)$ as the 
set of vertices reachable from $v$ by a path in $G$. For each 
$i\in\{1,\ldots,k\}$ the set~$Y_{y_i}$ consists of $y_i$, and, for 
$j\in\{0,\ldots,k\}$ with~$j\neq i$, each $w_j$ with $w\in V_i$ and the 
terminals~$\ell_{ij}\in L$. In particular, the sets $Y_{y_i}$ are disjoint and 
also partition the terminal set $L$. The set $Y_{r_1}$ consists of $r_1$ and the 
union $\bigcup_i Y_{y_i}$. Hence in order for $r_1$ to be connected to every 
terminal of $L$ in~$N$, for each $i\in\{1,\ldots,k\}$ the solution needs to 
include the edge $r_1y_i$ and at least one edge~$y_iw_0$ for some~$w\in V_i$. 
Since a terminal $\ell_{ij}$ is adjacent to the $j$-th copy $w_j$ of every 
vertex~$w\in V_i$, for each $j\neq i$ at least one edge $w_0w_j$ (for various 
$w\in V_i$) and a corresponding edge $w_j\ell_{ij}$ must be included in~$N$. 
These edges contribute a cost of $2k+2k(k-1)$ to $N$. 

Now consider the reachability set $Y_{z_{ij}}$ for some $1\leq i<j\leq k$. It 
consists of $z_{ij}$, all $z_e$ with~$e\in E_{ij}$, the $j$-th copy $w_j$ of 
every vertex $w\in V_i$ incident to edges of $E_{ij}$, the $i$-th copy $w_i$ of 
all vertices $w\in V_j$ incident to edges of $E_{ij}$, and corresponding 
terminals $\ell_{ij}$ and~$\ell_{ji}$. Since all terminals of $L$ are reachable 
from $r_2$ and the sets $Y_{z_{ij}}$ are disjoint, the sets $Y_{z_{ij}}$ 
partition~$L$. The set $Y_{r_2}$ consists of $r_2$ and the union $\bigcup_{i<j} 
Y_{z_{ij}}$, and so for every $1\leq i<j\leq k$ the solution $N$ must contain 
the edge $r_2z_{ij}$ and at least one edge $z_{ij}z_e$ for some $e\in E_{ij}$. 
In order for $r_2$ to connect to $\ell_{ij}$ in~$N$, the solution must also 
contain the edge $z_ew_j$ for some $w\in V_i$ incident to $e\in E_{ij}$. 
Analogously, the solution must also contain the edge $z_ew_i$ for $r_2$ to 
reach $\ell_{ji}$ in $N$ for some $w\in V_j$ incident to some $e\in E_{ij}$. 
These edges contribute a cost of $4{k\choose 2}$ to $N$.

Since all these necessary edges in $N$ sum up to a cost $2k+2k(k-1)+4{k \choose 
2}=4k^2-2k$, they are also the only edges present in~$N$. In particular, for 
each $i\in\{1,\ldots,k\}$ the solution contains exactly one edge $y_iw_0$ for 
some~$w\in V_i$, and therefore also must contain the $2(k-1)$ corresponding 
edges $w_0w_j$ and $w_j\ell_{ij}$ for~$j\neq i$. On the other hand, for every 
$1\leq i<j\leq k$ the solution contains exactly one edge $z_{ij}z_e$ for 
some~$e\in E_{ij}$, and therefore also must contain the corresponding edge 
$z_ew_j$ for the incident vertex $w\in V_i$ to $e$ and the corresponding edge 
$z_ew_i$ for the incident vertex $w\in V_j$ to $e$. Hence the solution $N$ 
corresponds to a subgraph of the instance of \pname{Multicoloured Clique} with 
$k$ pairwise adjacent vertices, i.e., it is a clique $K$ of size $k$.
\end{proof}

The reduction for the case when the pattern is a flawed $\alpha$-diamond is 
essentially the same as the one for pure $\alpha$-diamonds, as we show next.

\begin{lem}\label{lem:red-fdiam}
If $\mc{H}$ is the class of all flawed out-diamonds, then $\mc{H}$-DSN is 
\textup{W[1]}-hard parameterized by the number of terminals. The same holds if 
$\mc{H}$ is the class of all flawed in-diamonds.
\end{lem} 
\begin{proof}%
We only describe the case when $H$ is an out-diamond, as the other case is 
symmetric. The reduction builds on the one given in \autoref{lem:red-pdiam}: we 
simply add the additional terminal $x$ of $H$ to $G$, and connect it to $r_1$ 
and $r_2$ in $G$ by edges $xr_1$ and $xr_2$ with cost $1$ each. Given a clique 
of size $k$ in an instance to \pname{Multicoloured Clique}, consider the network 
$N$ in $G$ of cost $4k^2-2k$ suggested in \autoref{lem:red-pdiam}. We add the 
edges $xr_1$ and $xr_2$ to $N$, which results in a solution of cost $4k^2-2k+2$ 
for the flawed $\alpha$-diamond~$H$. On the other hand, any solution to $H$ must 
contain a path from $x$ to $r_1$ and from $x$ to~$r_2$. Since there is no path 
from $r_1$ to $r_2$, nor from $r_2$ to $r_1$ in the constructed graph~$G$, any 
solution to $H$ must contain both the edge $xr_1$ and the edge $xr_2$. Thus the 
minimal cost solution to $H$ in $G$ has cost $4k^2-2k+2$ and corresponds to a 
clique of size $k$ in the \pname{Multicoloured Clique} instance, as argued in 
the proof of \autoref{lem:red-pdiam}.
\end{proof}

Given the three reductions above, we can now prove \autoref{thm:main-hard}, 
based on the additional reduction given in \autoref{lem:closure}. 

\begin{proof}[Proof (of \autoref{thm:main-hard})]
  Let $\mc{H}'$ be the closure of $\mc{H}$ under identifying vertices
  and transitive equivalence. By assumption, $\mc{H}$ is not in
  $\mc{C}^*_{\lambda,\delta}$ for any $\lambda$ and $\delta$, and this
  is also true for the superset $\mc{H}'$ of $\mc{H}$. Thus
\autoref{thm:char} implies that $\mc{H}'$ fully contains one of five
classes: the class of all directed cycles, pure in-diamonds, pure
out-diamonds, etc. Suppose for example that $\mc{H}'$ contains the
class of all directed cycles, which we will denote by $\mc{H}''$. By
\autoref{lem:red-SCC}, we know that $\mc{H}''$-DSN is W[1]-hard and
$\mc{H}''$ is obviously decidable. Thus we can invoke
\autoref{lem:closure} to obtain that there is a parameterized
reduction from $\mc{H}''$-DSN to $\mc{H}$-DSN, and hence we can conclude
that the latter problem is also W[1]-hard. The proof is similar in the
other cases, when $\mc{H}'$ contains, e.g., every pure in-diamond or
every flawed in-diamond: then we use \autoref{lem:red-pdiam} or
\autoref{lem:red-fdiam} instead of \autoref{lem:red-SCC}.
\end{proof}

\printbibliography

\end{document}